\documentclass[12pt]{article}

\usepackage{amsthm}
\usepackage{amsmath}
\usepackage{amssymb}
\usepackage{amsthm}
\usepackage{bm}
\usepackage[round, authoryear]{natbib}
\usepackage[linesnumbered,ruled,vlined]{algorithm2e}
\usepackage{url}

\usepackage{xargs}                      % Use more than one optional parameter in a new commands
\usepackage[pdftex,dvipsnames]{xcolor}  % Coloured text etc.
\usepackage{float} %figure inside minipage
\usepackage{graphicx,wrapfig,lipsum}
\usepackage{xcolor}
\usepackage{calc} %scale svg image

\usepackage[colorinlistoftodos,prependcaption,textsize=normalsize]{todonotes}

\usepackage{mwe} % For dummy images
\usepackage{subcaption}
\usepackage{graphicx}
\usepackage{stmaryrd}
\usepackage{amsmath,bm}

\usepackage[most]{tcolorbox} %FOR COLOR BOX ENVIRONMENT
\usepackage{tikz}
\usetikzlibrary{decorations.pathreplacing,calligraphy}
\newcommandx{\unsure}[2][1=]{\todo[linecolor=red,backgroundcolor=red!25,bordercolor=red,#1]{#2}}
\newcommandx{\change}[2][1=]{\todo[linecolor=blue,backgroundcolor=blue!25,bordercolor=blue,#1]{#2}}
\newcommandx{\info}[2][1=]{\todo[linecolor=OliveGreen,backgroundcolor=OliveGreen!25,bordercolor=OliveGreen,#1]{#2}}
\newcommandx{\improvement}[2][1=]{\todo[linecolor=Plum,backgroundcolor=Plum!25,bordercolor=Plum,#1]{#2}}
\newcommandx{\thiswillnotshow}[2][1=]{\todo[disable,#1]{#2}}

\newcommand{\I}{\mathbf{1}}

\SetKwInput{KwInput}{Input}                % Set the Input
\SetKwInput{KwOutput}{Output}

\newtheorem{Definition}{Definition}[section]
\newtheorem{Theorem}{Theorem}[section]
\newtheorem{Lemma}{Lemma}[section]

\newtheorem{Proposition}{Proposition}[section]
\newtheorem{Remark}{Remark}[section]

\newtcolorbox[auto counter]{summary}[1][]{title={\bfseries Comment to settle~\thetcbcounter},enhanced,drop shadow={black!50!white},
  coltitle=black,
  top=0.3in,
  attach boxed title to top right=
  {xshift=0pt},
  boxed title style={size=small,colback=pink},#1}

\title{Hybrid of node and link communities for graphon estimation}
\author{Arthur Verdeyme\thanks{
    This work was supported by the European Research Council under Grant CoG 2015-682172NETS, within the Seventh European Union Framework Program.}\\
    Institute of Mathematics, Ecole Polytechnique Fédérale de Lausanne,\\
    Station 8, 1015 Lausanne, Switzerland\\
    and \\
    Sofia Olhede \\
    Institute of Mathematics, Ecole Polytechnique Fédérale de Lausanne,\\
    Station 8, 1015 Lausanne, Switzerland}\date{\today}

\begin{document}

\maketitle

\begin{abstract}
    Networks serve as a tool used to examine the large-scale connectivity patterns in complex systems.
    Modelling their generative mechanism nonparametrically is often based on step-functions, such as the stochastic block models.
    These models are capable of addressing two prominent topics in network science: link prediction and community detection.
    However, such methods often have a resolution limit, making it difficult to separate small-scale structures from noise.
    To arrive at a smoother representation of the network's generative mechanism, we explicitly trade variance for bias by smoothing blocks of edges based on stochastic equivalence.
    As such, we propose a different estimation method using a new model, which we call the \textit{stochastic shape model}.
    Typically, analysis methods are based on modelling node or link communities. In contrast, we take a hybrid approach, bridging the two notions of community.
    Consequently, we obtain a more parsimonious representation, enabling a more interpretable and multiscale summary of the network structure.
    By considering multiple resolutions, we trade bias and variance to ensure that our estimator is rate-optimal.
    We also examine the performance of our model through simulations and applications to real network data.\end{abstract}

\noindent%
\vfill

\newpage
%@@@@@@@@@@@@@@@@@@@@@@@@@@@@@@@
%@@@@@@@@@@@@@@@@@@@@@@@@@@@@@@@
%@@@@@@@@@@@@@@@@@@@@@@@@@@@@@@@
\section{Introduction}
\label{sec:intro}

This paper introduces a new method of estimating the generating mechanism of a network non-parametrically. Classically methods have been based on the stochastic equivalence of nodes, and corresponded to fitting a stochastic blockmodel~\citep{lancichinetti2009detecting}, or variants thereof. However, such models can create highly variable estimators, and require additional smoothing \citep{airoldismoothing2014, sischka2022stochasticsmoothing, li2022smoothinggraphon}. Our approach addresses this outstanding gap in current network methodology. To solve this problem, we introduce a model that bridges node and link communities, called the stochastic shape model, and use it as a rationale for smoothing the connectivity matrix. By doing so, we provide a simpler estimator whose performance does not suffer from this choice of bias-variance trade-off. 

Estimating the generating mechanism of a network is a well-studied problem. With the increased availability of large-scale network data sets, \textit{graphons} have emerged as a natural way to describe the generating mechanism of a network, assuming permutation invariance of the nodes~\citep{lovasz2012large,borgs2017graphons}. Such an assumption is not very restrictive as, in many applications, the ordering of nodes is arbitrary. Such an idea is summarized by the Aldous-Hoover theorem, \citep{hoover1979relations, aldous1981representations} from which one can show that a graphon can represent many latent variable graph models~\citep{orbanz2014bayesian,hoff2007modeling}. Standard practices to approximate a graphon either use a stochastic block model where the number of blocks depends on the number of nodes in the network, as do~\citet{olhede2014network, gao2015rate,chan2014consistent,klopp2017oracle}, or by using another non-parametric approximating technique \citep{chatterjee2015matrix, zhang2017estimating}.

Approaching such an estimation with stochastic block models is a good starting point for understanding node invariances~\citep{chen2020targeted}. However, we argue that this view is node-centric and can be restrictive to pattern such as overlapping or hierarchical communities~\citep{ho2012multiscale,latouche2016variational,li2022hierarchical} as well as heavily edge-focused network such as communication or collaboration networks. These ideas have found implementation in past years. \citet{evans2009line, ahn2010link} have introduced the idea of \textit{link communities}, to model interactions from edge similarity rather than nodal groups. In recovering the generating mechanism of a network, the work of \citet{crane2016edge} has led to the first graphon-like representation of network assuming edge-exchangeability, leading to models that could encapsulate even more patterns such as sparsity and power-law distribution~\citep{veitch2019sampling,dempsey2022hierarchical,zhang2022node}.

To bridge those two views, we provide a method which estimates a model based on the stochastic block model, which focuses on edge variables instead of nodes or realized edges. We call it the \textit{stochastic shape model}. This model focuses on a new notion of link community, one based on stochastic equivalence of edge variables groups, much like the mixed membership model from \citet{airoldi2008mixed}. This model can be thought of as a \textit{hybrid} between node and link communities, a term borrowed from \citet{he2015identification}.

To do so, we use a multiscale (but non-hierarchical) estimator, which is obtained as follows. First, we get a stochastic block approximation of the network, from \citet{olhede2014network}. As this creates a noisy estimator, we post-hoc smooth the non-parametric representation of the non-parametric estimate of the graphon. This smoothing is performed over a set of edges variables, thus creating link communities, from node communities, producing a shape in the space of edge variables that is not represented as a Cartesian product over sets of nodes. These shapes refer to the stochastic shape model that we introduce and can be seen as based on level set from \citet{osher2003level}. As we shall have a growing number of nodes, we will be able to approximate constant probabilities over arbitrary domains, where, for non-parametric approximation, the domains are assumed to shrink relative to the area of the graphon, i.e. $[0,1]^2$.
Note that we cannot do that fit to raw data (i.e. edge variables) directly, as then we would need to learn, from a graph with $n$ nodes, $\tbinom{n}{2}$ latent variables from $\tbinom{n}{2}$ observations, which would break a union bound necessary to determine the concentration of the fitted model. At best, the error rate would be constant, not decreasing in $n$~\citep{gao2015rate}. This is why we need to first have an initial estimate based on the stochastic block model.

We show the performance of our method both theoretically and empirically. We show that our estimator is rate-optimal both when the graphon is a stochastic shape model and a Hölder-smooth function using results and techniques from \citet{gao2015rate}. Our experiments were performed on both synthetic and real world data. On synthetic data, we used various type of graphons to show the versatility of our estimator. We also illustrate empirically our theoretical results such as the behavior of our considered loss functions, and the potential exponential reduction of parameters compared to stochastic block model methods such as \citet{wolfe2013nonparametric, olhede2014network, gao2015rate}. We maintain a similar, if not better, predictive performance. This is particularly of interest in community detection and also allows tackling a fundamental problem when inferring a graphon. As discussed in \citet{klimm2022modularity}, inference using a graphon model can trade a large complex network for a large complex object. This goes particularly against principles of community detection, as one aims to get a simpler summary of the networks patterns. Thus, smoothing and the potentially induced exponential reduction of parameters is of deep interest in building more interpretable summaries of networks. We illustrate this point on real-world data set.

In section 2, we describe our main results and modelling assumptions. Section 3 gives a more precise description of the theoretical framework, as well as the necessary theorems to obtain rate optimality. We further discuss, in section 4, how to build, in practice, our estimator and how this method build a new perspective for community detection. Finally, we illustrate the predictive performance of our method, in section 5, with both synthetic and real-world data sets.
%@@@@@@@@@@@@@@@@@@@@@@@@@@@@@@@
%@@@@@@@@@@@@@@@@@@@@@@@@@@@@@@@
%@@@@@@@@@@@@@@@@@@@@@@@@@@@@@@@

\section{Modelling assumptions and main results}

\subsection{Latent variable models}
A network can be represented by a $n\times n$ interaction data matrix $A$, referred to henceforth as an ``adjacency matrix'', whose $(i,j)$th entry represent the interaction between node $i$ and node $j$ via the absence or presence of an edge. In this paper, we consider an undirected, unweighted and without self-loops graph of $n$ nodes. That is to say that the connectivity can be encoded by an adjacency matrix $\left\{A_{i j}\right\}$ taking values in $\{0,1\}^{n \times n}$ such that $A_{ii} =0$ $\forall i \in [n]$ and $A_{ij}=A_{ji}$ $\forall i\neq j$ and $i,j\in[n]$. The value of $A_{i j}$ represents the presence or absence of an edge between the $i$th and $j$th nodes. The model, in this paper, of the edge variable $A_{ij}$ is $A_{i j}\sim \operatorname{Bernoulli}\left(\theta_{i j}\right)$ for $1 \leq j<i \leq n$, where $0\leq \theta_{ij}\leq 1$. This is the realization of $\binom{n}{2}$ independent Bernoulli trials. In such a setting, one can invoke Aldous-Hoover's theorem \citep{hoover1979relations,aldous1981representations, diaconis2007graph} of jointly exchangeable arrays as follows.
\begin{Theorem}[Aldous-Hoover]
\label{aldoushoover}
A random array $\left\{A_{i j}\right\}$ is jointly exchangeable if and only if it can be represented as follows: There is a random function $f:[0,1]^2 \rightarrow [0,1]$ such that
$$
A_{i j} \mid \xi_i, \xi_j \sim \operatorname{Bernoulli}\left(f\left(\xi_i, \xi_j\right)\right),
$$
where $\left(\xi_i\right)_{i \in \mathbb{N}}$ is a sequence of i.i.d $U[0,1]$ random variables, which are independent of $f$.
\end{Theorem}
This function $f$ is commonly referred to as a \textit{graphon}, or, equivalently, a graph limit function \citep{lovasz2012large}. The graphon is a non-negative symmetric bivariate function that allows us to represent a discrete network and a discrete set of probabilities to a continuous limiting object that lies in $[0,1]^2$. 
This concept plays a significant role in network analysis. Since the graphon is an object independent of the network number of nodes $n$, it gives a natural path to compare networks of different sizes. Moreover, model based prediction and testing can be done with the graphon framework \citep{lloyd2012random}. Besides non-parametric models, various parametric models have been proposed on the connectivity matrix $\left\{\theta_{i j}\right\}$ to capture different aspects of the network \citep{athreya2017statistical}.
Thus, we further assume the following graphon model, conditional on the latent variable $\bm{\xi}$, we set
\begin{equation}
\label{graphon}
    \theta_{i j}=f\left(\xi_i, \xi_j\right), \quad i \neq j \in[n] .
\end{equation}

Following Aldous-Hoover's theorem, we assume that the latent sequence $\left\{\xi_i\right\}$ are random variables sampled from a uniform distribution supported on $[0,1]^n$.
One can consider two estimation methods. A first one is estimating the point-wise probabilities, i.e. estimating $f(\xi_i,\xi_j)$ at specific point $(\xi_i,\xi_j)$, that we call in this paper \textit{value estimation}. The second one is a \textit{function estimation}, where one aims to recover the whole generating mechanism of the continuous graphon over its whole domain $[0,1]^2$.
Given $\left\{\xi_i\right\}$, we assume $\left\{A_{i j}\right\}$ are independent for $1 \leq j<i \leq n$. 
In the model \eqref{graphon},
as $\left\{\left(\xi_i, \xi_j\right)\right\}$ are 
latent random variables, $f$ can only be estimated from the response $\left\{A_{i j}\right\}$. 
Consequently, in a value estimation framework, this causes an identifiability problem, because without observing the latent variables, there is no way to associate the value of $f(x, y)$ with $(x, y)$. 
In this paper, we consider the following loss function
\begin{align}
\label{loss}
\frac{1}{n^2} \sum_{i, j \in[n]}\left(\hat{\theta}_{i j}-\theta_{i j}\right)^2,  
\end{align}
where $\hat{\theta}$ is the estimator of $\theta$.
Even without observing the design $\left\{\left(\xi_i, \xi_j\right)\right\}$, it is still possible to estimate the matrix $\left\{\theta_{i j}\right\}$ by exploiting its underlying structure modelled by (\ref{graphon}). 

The exchangeability assumption
implies that a graphon representation defines an equivalence class up to a measure-preserving transformation.
Thus, in a function estimation framework, one needs a metric that must reflect it. Consequently, we adapt the loss function \eqref{loss} into a mean integrated square error \citep{wolfe2013nonparametric} such as:
\begin{align}
\label{MISE}
    \inf _{\sigma \in \mathcal{M}} \iint_{(0,1)^2}\left|f(\sigma(x), \sigma(y))-\hat{f}(x, y)\right|^2\; d x\; d y,
\end{align}
where $\mathcal{M}$ is the set of all measure-preserving bijections of the form $\sigma:[0,1]\to [0,1]$ and $\hat{f}$ is the estimator of $f$. This defines a metric on the quotient space of graphons, as shown in \citet{lovasz2012large}. While one can find more details in \citet{tsybakov_2010}, it is, in general, impossible to recover a measurable function from a finite sample. However, this can be possible if we add further assumptions on the graphon function. 
As described in \citet{cai2014iterative}, recovering a function from its values at a random and finite set of inputs, under various assumptions on the function, is treated separately from the estimation of these values in many cases.
To perform such a task, we make the common assumption that the function is Hölder-continuous as in \citet{wolfe2013nonparametric, gao2015rate} and \citet{klopp2017oracle}. Indeed, since any estimator $\hat{f}$, based on a block representation, may be thought of as a Riemann sum estimate of $f$, we need to know when these sums converge. A bounded graphon on $[0,1]^2$ is Riemann-integrable if and only if it is almost everywhere continuous, following Lebesgue's condition. Assuming $(\alpha,M)$-Hölder continuity, where $\alpha$ is the Hölder coefficient and $M$ is the Hölder bounding constant, we have that
\begin{equation}
\label{holder_cond}
f \in \mathcal{H}(\alpha,M) \Leftrightarrow \sup _{(x, y) \neq\left(x^{\prime}, y^{\prime}\right) \in(0,1)^2} \frac{\left|f(x, y)-f\left(x^{\prime}, y^{\prime}\right)\right|}{\left|(x, y)-\left(x^{\prime}, y^{\prime}\right)\right|^\alpha} \leq M<\infty,
\end{equation}
where $\mathcal{H}(\alpha,M)$ is the class of $(\alpha,M)$-Hölder continuous functions and $f$ is assumed to be uniformly continuous, so that Riemann sums can be used to adjust its approximation error.

While using a H\"older graph limit is standard in non-parametric statistics, as it is unfeasible to estimate a full function as a graph limit, it is standard to use a block constant function to approximate the graph limit. We therefore define a block constant function to be %\arthur{Isn't it an intersection ?}\sofia{nope. this gives us symmetric probability surfaces, we talked about it...}
\begin{equation}
 f(x,y)=\sum_{a,b} \theta_{ab} \I((x,y)\in \omega_{ab}\cup (y,x)\in \omega_{ab}),    
\end{equation}
for $0\leq x<y\leq 1$, where we define $\omega_{ab}=\{(x,y):\; (a-1)/k\leq x<a/k,\; (b-1)/k\leq y<b/k\}$. Using the set of boxes $\omega_{ab}$, we can split $[0,1]^2$ into a checkerboard pattern. 

\subsection{The stochastic shape model}

We already know that block constant functions can be used to approximate an arbitrary Hölder smooth function~\citep{wolfe2013nonparametric, gao2015rate, klopp2017oracle}. We now assume more general methods. This will come via the notion of a level-set. A block constant function will have level sets corresponding to the blocks in the function. If we now want to group together edges for an arbitrary shape, rather than a square, the level set becomes a natural tool to use. We define the closed region ${\cal R}_c$ where the function $f(x,y)$ takes the value $\theta_c$. If we assume there are only $C$ distinct regions for the function $f(x,y)$ then the following set-up will work.
We now define the region constant function to be functions of the form
\begin{equation}f(x,y)=\sum_{c=1}^C \theta_c \I\left\{(x,y)\in {\mathcal{R}}_c\right\}.\label{unionoflevel}\end{equation}
We assume that $\cup_{c=1}^C {\cal R}_c=[0,1]^2$, and ${\mathcal{R}}_c\cap {\mathcal{R}}_{c'}=\emptyset$ for all $c\neq c'\in [C]$. The regions have to be chosen to respect the symmetry of the graphon. To solve the problems of value and function estimation, we consider the $\{\theta_{ij}\}_{1\leq i,j \leq n}$ to be from an arbitrary shape, where averaging of edge variables would be a natural estimation procedure.
\begin{Definition}[Stochastic Shape Model ($SSM$)]\label{def:StochShape_origin}
Assume we have defined $s\in{\mathbb{N}}^+$ non-intersecting closed regions in ${\cal T}= [0,1]^2\cap \{ x\leq y\}$, let us call them $S_c$ for $c\in[s]$. Define $S_{s}= {\cal T}\backslash \left\{\cup_{c<s} S_c\right\}$. We can then define the function $f$ for the $s$ constants $0<\theta_c<1$, for $c\in[s]$ to be
\begin{equation}
\label{ssm_origin}
f(x,y)=\left\{ \begin{array}{lcr}
\theta_c & {\mathrm{if}} & (x,y)\in S_c\\
\theta_c & {\mathrm{if}} & (y,x)\in S_c
\end{array} 
\right.
.
\end{equation}
\end{Definition}
We denote the parameter space for $\left\{\theta_{i j}\right\}$ by $\Theta_{s}\in [0,1]^s$, where $s$ is the number of shapes in the stochastic shape model. The exact definition of $\Theta_{s}$ is given in the next section.
The value of $\theta_{i j}$ only depends on the $p$th shape that the $(i,j)$th edge variable belong to.
As discussed in the introduction, one has to be careful with such a fluid model as, without further modelling restrictions, we can end up estimating $\binom{n}{2}$ parameters from $\binom{n}{2}$ variables. Thus, we need to keep the underlying nodal structure as in a stochastic block model. As such, we define a block-like stochastic shape model with $s$ shapes and a \textit{block-resolution} $k$, call it $SSM(s,k)$ as follows.
\begin{Definition}[$(s,k)$-Stochastic Shape Model ($SSM(s,k)$)]\label{def:StochShape}
Assume we have defined $s\in{\mathbb{N}}^+$ symmetric regions in ${\cal T}= [0,1]^2$ that are unions of blocks of length $k^{-1}$. For each $(\xi_i,\xi_j)\in[0,1]^2$ we define $w(\xi_i,\xi_j) = u(z(\lceil k\xi_i\rceil), z(\lceil k\xi_j\rceil))$ where $z: [n]\to[k]$ is the mapping from a node to its associated block and $u:\left[k^2\right]\to [s]$ maps a block to the shape it belongs. We can then define the function $f$ for the $s$ constants $0<\theta_c<1$, for $c\in[s]$ to be
\begin{equation}
\label{ssm_block}
f(\xi_i,\xi_j)=\left\{ \begin{array}{lcr}
\theta_c & {\mathrm{if}} & w(\xi_i,\xi_j)=c\\
\theta_c & {\mathrm{if}} & w(\xi_i,\xi_j)=c
\end{array} 
\right.
.
\end{equation}
\end{Definition}
To the best of our knowledge, it is the first time that a model such as the $SSM(s,k)$ is considered as both a data generating mechanism and model for predictions. 

We derive rates of convergence for our method using a similar approach than \citet{gao2015rate} and \citet{klopp2017oracle} for both upper and lower bounds. We use properties of the packing number of possible shape assignment, which further allow us to obtain better rate of convergence as well as a significant reduction of the number of parameters compared to methods based on stochastic block models, both theoretically and in practice. Indeed, as explained by \citet{gao2015rate}, the packing number helps characterize our ignorance of the model, whether of the graphon latent variables structures or the number of shapes. This last point gives an intuition on the improvement brought by considering the stochastic shape model instead of its block counterpart, as the shapes provide more flexibility. Thus, we reduce the variance brought, at the price of introducing bias, by a stochastic block approximation in a bias-variance tradeoff setting from the choice of the loss function defined in \eqref{loss}.
Our estimator is constructed by using a stochastic block estimator, with $k$ blocks that we further group together based on stochastic equivalence. Consequently, we obtain $s$ shapes. This method of estimation requires that $n>k^2-s$ and there is an invertible mapping from the shapes to the blocks. This is further described in Section 3 and illustrated in Figure \ref{fig:framework_rpz}.
As such, in a setting where $\{\theta_{ij}\}$ is sampled from a stochastic shape model, the minimax rate is summarized by the following.
\begin{Theorem}
\label{thm:ssm}
    Assuming the $(s,k)$-stochastic shape model, we have 
    $$\inf _{\hat{\theta}} \sup _{\theta \in \Theta_s} \mathbb{E}\left\{\frac{1}{n^2} \sum_{i, j \in[n]}\left(\hat{\theta}_{i j}-\theta_{i j}\right)^2\right\} = \Theta\left(\frac{s}{n^2}+\frac{\log (\operatorname{max}(k,s))}{n}\right),$$
    for any $s\in [n^2]$ with $n>\operatorname{max}(0,k^2-s)$.
\end{Theorem}
Each term of this theorem can be interpreted as in \citet{gao2015rate}. That is, the variability $\frac{s}{n^2}$ illustrates the error induced by estimating $s$ parameters with $n^2$ observations. The clustering part, from $\frac{\operatorname{log}(\operatorname{max}(k,s))}{n}$, represents the identifiability mentioned previously. 
Similarly, when $\{\theta_{ij}\}$ is sampled from a graphon $f \in \mathcal{H}(\alpha,M)$, we obtain the following
\begin{Theorem}
\label{thm:holder}
    Assume $f \in \mathcal{H}(\alpha,M)$. We have
    $$
    \inf _{\hat{\theta}} \sup _{f \in \mathcal{F}_\alpha(M)} \mathbb{E}\left\{\frac{1}{n^2} \sum_{i, j \in[n]}\left(\hat{\theta}_{i j}-\theta_{i j}\right)^2\right\} = \Theta\left(n^{-2 \alpha /(\alpha+1)}+ \frac{\log n}{n}\right),
    $$
    where the expectation is jointly over $\left\{A_{i j}\right\}$ and $\left\{\xi_i\right\}$.
\end{Theorem}
Here, the nonparametric rate illustrates the error induced by approximating a Hölder-smooth function by a stochastic shape model. As in \citet{gao2015rate, olhede2014network}, the nonparametric rate dominates when $\alpha<1$.
To illustrate this bias-variance trade-off and the performance of our estimator, we will compare our method to state-of-the-art methods such as the Sorting and Smoothing from \citet{chan2014consistent}, the Universal Singular Value Threshold from \citet{chatterjee2015matrix} and the Network histogram from \citet{olhede2014network}.
%@@@@@@@@@@@@@@@@@@@@@@@@@@@@@@@
%@@@@@@@@@@@@@@@@@@@@@@@@@@@@@@@
%@@@@@@@@@@@@@@@@@@@@@@@@@@@@@@@
\section{Estimation and performance}
\label{sec:meth}

\subsection{Methodology}
\label{sec:meth:oracle}

%---- Framework ----

We, here, propose an estimator for both the stochastic shape model and the non-parametric graphon estimation under Hölder smoothness. To introduce the estimator, let us define the set $\mathcal{Z}_{n, k}=\{z:[n] \rightarrow[k]\}$ to be the collection of mappings from $[n]$ to $[k]$ with some integers $n$ and $k$, defining a clustering on the nodes. Define also $\mathcal{U}_{k,s} = \{u:[k]^2 \rightarrow[s]\}$ to be the collection of mappings from $[k^2]$ to $[s]$, defining a clustering on the blocks.
Given $z \in \mathcal{Z}_{n, k}$, the sets $\left\{z^{-1}(a): a \in[k]\right\}$ form a partition of $[n]$, with $z^{-1}(a) \cap z^{-1}(b)=\emptyset$ for any $a \neq b \in[k]$. We further assume that the cardinality of a node group is constant, call it $h>1$. Given a matrix $\left\{\eta_{i j}\right\} \in \mathbb{R}^{n \times n}$ and a node-partition function $z \in \mathcal{Z}_{n, k}$, we define as in \citet{olhede2014network, gao2015rate} the $(a,b)$-block average, i.e. the edge density within block $(a,b)$, to be 
\begin{equation}
\label{block_average}
    \bar{\eta}_{a b}(z)=\frac{1}{h^2} \sum_{i \in z^{-1}(a)} \sum_{j \in z^{-1}(b)} \eta_{i j},
\end{equation}
for $a \neq b \in[k]$, and
\begin{equation}
\label{block_average_sym}
    \quad \bar{\eta}_{a a}(z)=\frac{1}{h\left(h-1\right)} \sum_{i \neq j \in z^{-1}(a)} \eta_{i j},
\end{equation}
for $a \in[k]$.

Next, define a map $u\in\mathcal{U}_{k,s}$ as a partition of the blocks in the stochastic shape sense, such that all clusters/tiles of similar density are grouped together. Now, given a matrix $\left\{\eta_{i j}\right\} \in \mathbb{R}^{n \times n}$, a node-partition function $z \in \mathcal{Z}_{n, k}$ and a tile partition function $u\in\mathcal{U}_{k,s}$, we define the shape average to be
\begin{equation}
    \label{est_hist_shape}
    \bar{\bar{\eta}}_{c}(w)=\frac{1}{\left|S_c\right|h^2} \sum_{(i,j) : w(i,j)\in S_c} \eta_{i j}=\frac{1}{\left|S_c\right|} \sum_{a,b\in S_c}\bar{\eta}_{ab},
\end{equation}
where $S_c$ is defined in Definition \eqref{def:StochShape_origin} and $w$ is a link community mapping, illustrated in Figure \ref{fig:framework_rpz}, such that $w(i,j):=u\circ(z_i\times z_j)$ as we aim to cluster edges into arbitrary regions. We will often denote the $w$ operator as $w\equiv u\circ z^2$ for simplicity of writing.

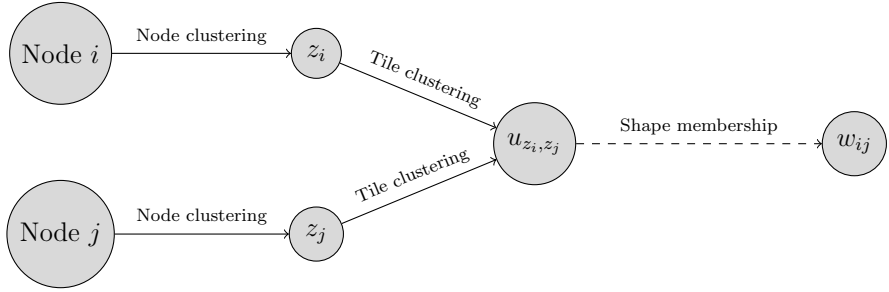
\begin{figure}[ht!]
    \centering
    \begin{tikzpicture}[node distance=2cm, auto, scale=0.85, transform shape]
      % First sequence
      \node[circle, draw, fill=gray!30] (uz) {$u_{z_i, z_j}$};
      \node[circle, draw, fill=gray!30] (wij) [right of=uz, xshift = 3cm] {$w_{ij}$};
      \path[dashed, ->] (uz) edge node[above, font=\scriptsize] {Shape membership} (wij);
    
      % Second sequence
    \node[circle, draw, fill=gray!30] (zi) [above left of=uz, xshift = -2cm] {$z_i$};
      \node[circle, draw, fill=gray!30] (nodei) [left of=zi, xshift = -2cm] {Node $i$};
      \path[->] (nodei) edge node[above, font=\scriptsize] {Node clustering} (zi) ;
      \path[->] (zi) edge node[above, font=\scriptsize, sloped] {Tile clustering} (uz);
    
      % Third sequence 
    \node[circle, draw, fill=gray!30] (zj) [below left of=uz, xshift = -2cm] {$z_j$};
      \node[circle, draw, fill=gray!30] (nodej) [left of=zj, xshift = -2cm] {Node $j$};
      \path[->] (nodej) edge node[above, font=\scriptsize] {Node clustering} (zj);
      \path[->] (zj) edge node[above, font=\scriptsize, sloped] {Tile clustering} (uz);
    \end{tikzpicture}
    \caption{Illustration of the shape membership operator, $w = u\circ z^2$.}
    \label{fig:framework_rpz}
\end{figure}

%----METHOD OF ESTIMATION----

Our proposed method aims to smooth a histogram representation based on nodal block averages, as defined in \eqref{block_average} and \eqref{block_average_sym}, to obtain a stochastic shape representation, based on \eqref{est_hist_shape}, in an attempt to reduce the variance of the original histogram. Assuming that $n=hk$ where $h$ is the width of the blocks, the shape average is the average of similar behaving blocks
$$ \bar{\bar{A}}_{c}(w)=\frac{1}{\left|S_c\right|} \sum_{(a,b)\in S_c} \bar{A}_{ab} = \frac{1}{\left|S_c\right|h^2} \sum_{(i,j) : w(i,j)\in S_c} A_{i j},$$
where $S_c := \left\{(a,b): \bar{A}_{ab} = \theta_c \right\}$. For any $Q=\left\{Q_{c}\right\} \in \mathbb{R}^{s}$, $z \in \mathcal{Z}_{n, k}$ and $u\in\mathcal{U}_{k,s}$, define the objective function
$$L(Q, w=u\circ z^2)=\sum_{c} \sum_{(a,b)\in u^{-1}(c)\times u^{-1}(c))}\sum_{(i,j)\in z^{-1}(a) \times z^{-1}(b)}\left(A_{i j}-Q_{c}\right)^2 .$$
For any optimizer
$$(\hat{Q}, \hat{w}=\hat{u}\circ \hat{z}^2) \in \underset{Q \in \mathbb{R}^{s}, u\in \mathcal{U}_{k,s}, z \in \mathcal{Z}_{n, k}}{\operatorname{argmin}} L(Q, u\circ z^2),$$
we define the estimator of $\theta_{ij}$ to be 
$$\hat{\theta}_{i j}=\hat{Q}_{\hat{w}(i,j)},$$
where $\hat{\theta}_{i j}=\hat{\theta}_{j i}$ for $i<j$ and $\hat{\theta}_{i i}=0$. The procedure can be understood as first clustering the data by an estimated $\hat{z}$, estimating the block averages, and then estimating the model parameters via averages of block averages. By the least squares formulation, it is easy to observe the following property.
\begin{Proposition}
\label{minsol}
For any minimizer $(\hat{Q}, \hat{w})$, the entries of $\hat{Q}$ has representation
$$
\hat{Q}_{c}=\bar{\bar{{A}}}_{c}(\hat{w}),
$$
for all $c \in[s]$.
\end{Proposition}
Using the least square criterion, this estimator is clustering with operation $\hat{u}\in\mathcal{U}_{k,s}$, using a density-similarity based argument, of a histogram approximation after finding the optimal community membership $\hat{z} \in \mathcal{Z}_{n, k}$. As such, we keep a histogram framework as an estimator. However, it is known that, in non-parametric regression, histograms cannot achieve optimal convergence for Hölder smooth functions with exponent $\alpha>1$ \citep{tsybakov_2010}. Nonetheless, following the work of \citet{gao2015rate} and \citet{wolfe2013nonparametric}, who showed rates of convergence under block models and non-parametric graphon settings such as Hölder continuity, we now are going to show optimal rates results assuming a stochastic shape estimator as defined in the proposition above, both in a value and function estimation. 

\subsection{Rates of convergence}
\label{sec:meth:rates}

\textit{An upper bound for an underlying $SSM(s,k)$ with $s$ shapes and $k$ block-resolution.} Estimating the rate of convergence, when the stochastic shape model is specified, follows a similar approach as \citet{wolfe2013nonparametric, gao2015rate}. We, here, borrow their terminology. As such, the error is composed of a clustering rate term, corresponding to the $\operatorname{log}(\operatorname{max}(k,s))/n$ term and a variance term which is $s/n^2$. This yields the following rates for a stochastic shape model with $s$ shapes.
 
\begin{Theorem}
\label{theorem2.1}
For any constant $C^{\prime}>0$, there is a constant $C>0$ only depending on $C^{\prime}$, such that
$$
\frac{1}{n^2} \sum_{ij}\left(\hat{\theta}_{ij}-\theta_{ij}\right)^{2} \leq C\left(\frac{s}{n^2}+\frac{\operatorname{log}(\operatorname{max}(k,s))}{n}\right),
$$
with probability at least $1-\exp \left(-C^{\prime} n \log s\right)$, uniformly over $\theta \in \Theta_{k}$. Furthermore, we have
$$
\sup _{\theta \in \Theta_{s}} \mathbb{E}\left\{\frac{1}{n^2} \sum_{ij}\left(\hat{\theta}_{ij}-\theta_{ij}\right)^{2}\right\} \leq C_{1}\left(\frac{s}{n^2}+\frac{\operatorname{log}(\operatorname{max}(k,s))}{n}\right),
$$
for all $s\in [n^2]$ with some universal constant $C_{1}>0$ and $n>\operatorname{max}(0,k^2-s)$.
\end{Theorem}
\iffalse
For a more in-depth characterization of these results, assume $s\sim n^\delta$ for $\delta\in[0,2]$. Taking all possible cases of $\delta$ to characterize in depth the rate, we thus have 
$$
\frac{s}{n^2}+\frac{\log s}{n} \leq \begin{cases}n^{-2}, & s=1, \\ n^{-1}, & \delta=0, s \geq 2, \\ n^{-1} \log n, & \delta \in(0,1], \\ n^{-2+\delta}, & \delta \in(1, 2] .\end{cases}.
$$
This fully describes the rates of convergence when estimating the stochastic shape model using our average of blocks averages method.
\fi
This is similar to the results from \citet{gao2015rate} (Theorem 2.1). This is not a surprise as, in non-parametric estimation, the rate of convergence follows the form 
$$\frac{(\#\text{parameters})+\log (\#\text{models})}{\#\text{samples}},$$
as described in \citet{tsybakov_2010, gao2021minimax}. Following this description, we have that $s/n^2$ correspond to the non-parametric rate while the $\operatorname{log}(\operatorname{max}(k,s))/n$ term correspond to the clustering rate of the stochastic shape model. This last term and the similarity of the rates between Theorem \ref{theorem2.1} and \citet{gao2015rate} should not come as a surprise, as we have that the number of shape parameters and block parameters behaves similarly if we assume the stochastic shape model to be a stochastic block model. It also follows that the different regimes of rates described just before behave similarly to \citet{gao2015rate} in the sense that the non-parametric rate dominates when $s$ is large and the clustering rate dictates the convergence when $s$ is small. 

\textit{A lower bound for an underlying $SSM(s,k)$ with $s$ shapes and $k$ block-resolution.} Proving that the lower bound's rate is the same as the upper bound asymptotically and up to a constant will show that the rate of convergence from Theorem \ref{theorem2.1} cannot be improved. Getting such minimax bounds, we would achieve rate optimality. As described in \citet{castillo2022uniform}, minimax bounds can be represented by a decreasing function $\tau(n)$ such that
$$\operatorname{sup}_{\theta\in\Theta_{s}} \mathbb{E}\left\{\frac{1}{n^2} \sum_{ij}\left(\hat{\theta}_{ij}-\theta_{ij}\right)^{2}\right\}\geq \tau(n),$$
over all estimators.
Thus, taking a subset of  the parameter space $\Theta_s$ will not increase the lower bound as we are considering the supremum. Consequently, it is enough to consider a subclass $\Theta_s^{'}\subset\Theta_s$ of this parameter subspace as a lower bound of this space will be lower than the supremum over the original set.
This is key for us to get lower bounds, as we can get them directly in \citet{gao2015rate}. Indeed, as SBMs is a sub-model of SSMs, any lower bound obtained assuming an underlying SBMs generating models with $k$ blocks can be adapted into an SSM generating model with a number of shapes $s=k^2$. As such, lower bounds from \citet{gao2015rate} and their underlying construction can be adapted to an SSM generating mechanism, call it $\Theta_{s, LB}$.
As follows, we provide a similar reasoning and discussion as \citet{klopp2017oracle} as to how they adapted the proof to their framework. The main difference from their proof is that now the degrees of freedom are smaller as we go from $k^2$ from $s$. Building on these insights, we obtain the following theorem, analogous to \citet{gao2015rate}.
\begin{Theorem}
\label{ssm:lwbnd}
    There exists a universal constant $C>0$, such that
    $$
    \inf _{\hat{\theta}} \sup _{\theta \in \Theta_s} \mathbb{P}\left\{\frac{1}{n^2} \sum_{i, j \in[n]}\left(\hat{\theta}_{i j}-\theta_{i j}\right)^2 \geq C\left(\frac{s}{n^2}+\frac{\operatorname{log}(\operatorname{max}(k,s))}{n}\right)\right\} \geq 0.8,
    $$
    and
    $$
    \inf _{\hat{\theta}} \sup _{\theta \in \Theta_s} \mathbb{E}\left\{\frac{1}{n^2} \sum_{i, j \in[n]}\left(\hat{\theta}_{i j}-\theta_{i j}\right)^2\right\} \geq C\left(\frac{s}{n^2}+\frac{\operatorname{log}(\operatorname{max}(k,s))}{n}\right),
    $$
    for any $s \in[n^2]$ and $n>\operatorname{max}(0,k^2-s)$.
\end{Theorem}
To obtain this result, we only modify the construction in \citet{gao2015rate}
of the connection probabilities of $\Theta_{s, LB}$ such that they are defined as $\frac{1}{2} + c \sqrt{\frac{\operatorname{log}(\operatorname{max}(k,s))}{n}}\omega_a$ with suitably chosen $\omega_a \in \{0,1\}$ and $c>0$ is small enough, similar to \citet{klopp2017oracle} adaptation.

\textit{Hölder-smooth graphons.}
As described in the previous section, the graphon estimation is also composed of a clustering rate term and a variance term of a similar form.
The assumption of $\mathcal{H}(\alpha, M)$-smoothness introduces another contribution to the error: the bias term.
As $f$ is continuous, it is Riemann integrable. Thus, any estimator $\hat{f}$ can be viewed as a Riemann approximation of $f$. As such, the bias term illustrates the error induced by estimating a continuous function by an histogram.
In a stochastic block framework, like \citet{olhede2014network, gao2015rate} or \citet{klopp2017oracle}, the bias is expressed as the distance of an \textit{oracle estimator} to the true function.
This oracle assumes access to additional information, making it an ideal estimator.
In the previously cited articles, this corresponds to a knowledge of the ordering of $\{\xi\}$, which, consequently, leads to assuming the estimator is a regular stochastic block model, i.e. blocks all have the same size \citep{olhede2014network}. In this resulting $k$-block stochastic block estimator, let $z^*$ be the ideal mapping from latent variables to their block memberships. Following this, one can show that the bound on the bias is depending on the maximum distance between latent variables $\{\xi_i\}$, call it $D$, the diameter. As the estimator is a regular stochastic block model with $k$ blocks, this leads to the following bound on the diameter
$$D=\operatorname{max}_{i,j\in [n]:z^*(i)=z^*(j)}\left(|\xi_i-\xi_j|\right)\leq k^{-1}$$
Adapting this to our shape framework means that we specify an oracle $w^*$, as the ideal mapping from pairs of latent variables to their shape memberships, and a shape diameter as the maximum distance between two pairs of latent variables $(\xi_i, \xi_j)$ and $(\xi_u,\xi_v)$ for $(i,j)\neq (u,v)$ belonging to the same shape.
As such, we set for $i,j,u,v\in[n]$
\begin{align}
\label{diamBeta}
D_{w^*} = \operatorname{max}_{(i,j)\neq (u,v):w^*(i,j)=w^*(u,v)}\left(|\xi_i-\xi_u|+|\xi_j-\xi_v|\right) \propto s^{-\frac{\beta}{2}},  
\end{align}
where $D$ is the biggest diameter of all shapes (i.e. its maximum distance between two points). Here, $\beta$ represents the degree up to which we do not know the relationship between the number of shapes $s$ and the diameter $D$, assuming an underlying stochastic shape model.
This shows how one should adapt modelling assumptions when going from a block perspective to a shape one.
The following lemma gives a bound on the bias term, i.e. with respect to the oracle $w^*$ shape assignment.

\begin{Lemma}
\label{lemma2.1}
There exists a constant $\beta>1$ and $w^{*} \in \mathcal{W}_{n, s}$, such that $D_{w^*}\propto s^{-\frac{\beta}{2}}$ and
$$
\frac{1}{n^2} \sum_{c\in[s]}\sum_{\left\{(i,j): w^{*}(i,j)=c \right\}}\left(\theta_{ij}-\bar  {\bar{\theta}}_{c}\left(w^{*}\right)\right)^{2} \leq C M^{2}\left(\frac{1}{s^{\beta}}\right)^{\alpha \wedge 1},
$$
holds for some universal constant $C>0$.
\end{Lemma}
Using this lemma and the result's framework from Theorem \ref{theorem2.1}, we have the following bound on the error 
$$\left(\frac{1}{s^\beta}\right)^{\alpha \wedge 1}+\frac{s}{n^2}+\frac{\operatorname{log}(\operatorname{max}(k,s))}{n},$$
up to a constant, depending on $M$ and $\alpha$. Here the bias-variance trade-off is clear for our method and illustrates our motivation to smooth further step-graphons based on SBM approximation. As in \citet{gao2015rate, klopp2017oracle}, we pick the best $s$ such that we obtain the classical rate of convergence of $\mathcal{H}(\alpha, M)$-smooth function \citep{tsybakov_2010}. Doing so results in the following theorem for smooth graphons.

\begin{Theorem}
\label{theorem2.3}
For $f\in \mathcal{H}(\alpha, M)$, there exists a constant $\beta>1$ such that for any $C'>0$, there exists a constant $C>0$ only depending on $C'$, $M$, $\alpha$ and $\beta$, where the following holds
    $$
    \begin{aligned}
    \frac{1}{n^2} \sum_{i j}\left(\hat{\theta}_{i j}-\theta_{i j}\right)^2
    &\leq C \left(n^{-2 \alpha /(\alpha+1)}+\frac{\operatorname{log}(n)}{n}\right),
    \end{aligned}
    $$
with probability at least $1-\exp \left(-C^{\prime} n\right)$, with $s = \left\lceil n^{\frac{2\beta^{-1}}{\alpha\wedge 1+1}}\right\rceil$. Furthermore,
$$
\sup _{f \in \mathcal{H}(\alpha, M)} \mathbb{E}\left\{\frac{1}{n^2} \sum_{i, j \in[n]}\left(\hat{\theta}_{i j}-\theta_{i j}\right)^2\right\} \leq C_1 \left(n^{-2 \alpha /(\alpha+1)}+\frac{\operatorname{log}(n)}{n}\right),
$$
for some other constant $C_1>0$ only depending on $M$. Both the probability and the expectation are jointly over $\left\{A_{i j}\right\}$ and $\left\{\xi_i\right\}$.
\end{Theorem}

\iffalse
\begin{Theorem}
\label{oldtheorem2.3}
Let $\gamma = \beta\left(\alpha\wedge 1\right)$ and $s = \left\lceil2\gamma n^{\frac{1}{\gamma +\frac{1}{2}}}\right\rceil$, then for any $C'$, there exists a $C>0$ only depending on $C$, $M$ and $\beta$ such that
    $$
\begin{aligned}
\frac{1}{n^2} \sum_{i j}\left(\hat{\theta}_{i j}-\theta_{i j}\right)^2
&\leq C \left(n^{\frac{-2\gamma}{\gamma+\frac{1}{2}}}+\frac{\operatorname{log}(n)}{n}\right),
\end{aligned}
$$
with probability at least $1-\exp \left(-C^{\prime} n\right)$, uniformly over $f\in \mathcal{H}(\alpha, M)$. Furthermore,
$$
\sup _{f \in \mathcal{F}_\alpha(M)}\mathbb{E}\left\{\frac{1}{n^2} \sum_{i, j \in[n]}\left(\hat{\theta}_{i j}-\theta_{i j}\right)^2\right\} \leq C \left(n^{\frac{-2\gamma}{\gamma+\frac{1}{2}}}+\frac{\operatorname{log}(n)}{n}\right),
$$
for some other constant $C_1>0$ only depending on $M$. Both the probability and the expectation are jointly over $\left\{A_{i j}\right\}$ and $\left\{\xi_i\right\}$.
\end{Theorem}
\fi
Similarly to Theorem \ref{theorem2.1}, the rate of convergence is divided into two: the nonparametric rate $n^{-2 \alpha /(\alpha+1)}$ and the clustering rate $\frac{\operatorname{log}(n)}{n}$. Following our previous discussion of the various behaviors under different shape regimes for Theorem \ref{theorem2.1}, we have for $\alpha\in (0,1)$ that the non-parametric rate dominates.

A direct consequence of this theorem when comparing to methods such as \citet{olhede2014network} and \citet{gao2015rate} is the number of parameters. Indeed, we note that the ratio of parameters between their methods and ours is of order $O\left(n^{2/{\alpha\wedge 1+1}}\right)$. This illustrates the advantages that using a stochastic shape model can bring, as it provides a simpler estimator to our network data. This is further showed by our experimental results, where we also showed an exponential reduction of parameters, see section \ref{sec:synth_data}.

For the lower bound of the error when assuming $f \in \mathcal{H}(\alpha, M)$ where $\alpha, M>0$, we use a similar argument as for Theorem \ref{ssm:lwbnd} and thus can use \citet{gao2015rate} result to obtain our lower bound, necessary to prove Theorem \ref{thm:holder}.

We just proved two theorems in a value estimation setting. We now extend it to function estimation. To make the difference between the two, we now use the following notation $f(\xi_i, \xi_j) = \theta_{ij}$, similarly for $\hat{f}$ and $\tilde{f}$.
Doing so, we change the metric we consider to the so-called \textit{cut distance} in the theory of graph limits \citep{lovasz2012large}. 
Consequently, we study the estimation of $f$ in the quotient space of graphons, that is, up to permutations in $\mathcal{M}$ of the node labels. In the previous section, we obtained convergence rates showing that the difference between $\hat{f}$ and $f$ shrinks to zero as $n \rightarrow \infty$ under the assumptions above. Similar to \citet{olhede2014network} and \citet{klopp2017oracle}, we consider the mean integrated square error, shorted to MISE, and take its greatest lower bound over all possible rearrangements $\sigma \in \mathcal{M}$ ; as it was shown in \citet{lovasz2012large} that this defines a metric on the quotient space of graphons. This gives
\begin{equation}
\label{mise}
    \operatorname{MISE}\left(\hat{f},f\right)=\mathbb{E} \inf _{\sigma \in \mathcal{M}} \iint_{(0,1)^2}\left|f\left(x, y\right)-\hat{f}\left(\sigma\left(x\right), \sigma(y)\right)\right|^2 d x d y .
\end{equation}
This definition selects the closest $\hat{f}$ in the equivalence class of graphons, given the unknown ordering of the data $A$ induced by $\left\{\xi_1, \ldots, \xi_n\right\}$ in the model of \eqref{graphon}. When such a procedure is applied to methods assuming a model such as the SBM, the location index is important and can be retrieves using the latent variables, as we can see in the following definition of the empirical graphon
$$\hat{f}_{\theta}\left(x,y\right) = \theta_{\lceil n x\rceil,\lceil n y\rceil},$$
where $\theta$ is a $n\times n$ matrix with entries in $[0,1]$. \citet{klopp2017oracle} proved some very useful results using a clever mix of covering number, Hoeffding's inequality and ordered statistics, allowing a direct adaptation to our work. For this, we need their following results.
\begin{Lemma}[\citet{klopp2017oracle}]
\label{lemmaklopp3.1}
    For any graphon $f$ in the space of graphons, and any estimator $\hat{\theta}$ of $\theta$ such that $\hat{\theta}$ is a $n \times n$ matrix with entries in $[0,1]$, we have
    $$
    \mathbb{E}\left[\operatorname{MISE}\left(\hat{f}_{\hat{\theta}}, f\right)\right] \leq 2 \mathbb{E}\left[\frac{1}{n^2}\left\|\hat{\theta}-\theta\right\|^2\right]+2 \mathbb{E}\left[\operatorname{MISE}\left(\hat{f}_{\theta}, f\right)\right].
    $$
\end{Lemma}
The triangle inequality from the previous lemmas allows us to combine results from the previous section to the following result to obtain the rates of convergence for function estimation. A direct consequence of this lemma is that graphon optimal estimation rates are slower than if we were to estimate the probability matrix associated to it. This is formalized by the following lemma, adapted from \citet{klopp2017oracle}.

\begin{Proposition}[\citet{klopp2017oracle}]
    Let $f \in \mathcal{H}(\alpha, M)$ where $\alpha, M>0$. Then
    $$
    \mathbb{E}\left[\operatorname{MISE}\left(\hat{f}_{\theta}, f\right)\right] \leq C {n^{-\alpha \wedge 1}},
    $$
    where the constant $C$ depends only on $M$ and $\alpha$.
\end{Proposition}

From Lemma \ref{lemmaklopp3.1}, we first have the bound from for the right term from the previous proposition and the bound for the left term follows from Theorem \ref{theorem2.3}. Thus, we obtain the following result.

\begin{Theorem}
\label{theorem_func}
For $f\in \mathcal{H}(\alpha, M)$, there exists a constant $\beta>1$ such that for any $C'>0$, there exists a constant $C>0$ only depending on $C'$, $M$, $\alpha$ and $\beta$, where the following holds
    $$
\begin{aligned}
\operatorname{MISE}\left(\hat{f}_{\hat{\theta}}, f\right)
&\leq C \left(n^{-2 \alpha /(\alpha+1)}+\frac{\operatorname{log}(n)}{n}+n^{-\alpha \wedge 1}\right),
\end{aligned}
$$
with probability at least $1-\exp \left(-C^{\prime} n\right)$, with $s = \left\lceil n^{\frac{2\beta^{-1}}{\alpha\wedge 1+1}}\right\rceil$. Furthermore,
$$
\sup _{f \in \mathcal{F}_\alpha(M)} \mathbb{E}\left[\operatorname{MISE}\left(\hat{f}_{\theta}, f\right)\right] \leq C_1 \left(n^{-2 \alpha /(\alpha+1)}+\frac{\operatorname{log}(n)}{n}+n^{-\alpha \wedge 1}\right),
$$
for some other constant $C_1>0$ only depending on $M$. Both the probability and the expectation are jointly over $\left\{A_{i j}\right\}$ and $\left\{\xi_i\right\}$.
\end{Theorem}

Similar to the value estimation framework, we can see here that, when $\alpha\geq 1$, the clustering rate dominates. However, we drift away from the value estimation results when $\alpha<1$ as the term $n^{-\alpha}$ dominates over the non-parametric rate $n^{-2\alpha/(\alpha+1)}$. This is the cost of doing function estimation instead of value, as one aims to further recover the whole sampling procedure by finding the permutation of nodes that minimizes the cut-distance. 
%@@@@@@@@@@@@@@@@@@@@@@@@@@@@@@@
%@@@@@@@@@@@@@@@@@@@@@@@@@@@@@@@
%@@@@@@@@@@@@@@@@@@@@@@@@@@@@@@@
\section{Approximation procedure}
\subsection{Determining the estimator}

In network analysis, some of the main challenges are recovering the generating mechanism of the network, for which graphon estimation is a part of, predicting links and detecting communities within the network. Methods of graphon estimation based on a stochastic block model span these problems.
This can be seen in \citet{olhede2014network} and was further developed in \citet{gao2015rate} where they related their graphon estimation framework to link prediction if only a part of the adjacency is observed. They also related their work to community detection following the work of \citet{lei2015consistency} and  \citet{chin2015stochastic} by providing results on the convergence of their method under the operator norm, which is used to derive misclassification error of spectral clustering.
However good such methods are, they often have trouble following the principle of Occam's razor: some methods obtaining sometimes 10 times more blocks than the theoretical block model, which is to be expected from likelihood, or $L_2$, maximization methods. As described in \citet{valles2018consistencies}, \textit{the most plausible model is not the most predictive}. By our framework, we provide a method that aim to be in the middle: selecting one of the most plausible models while having one of the best predictive performance. That is, performing model selection that balances specificity and generalizability \citep{copas1983regression}. 

To achieve this, we develop a heuristic framework in two steps. The first step constructs a non-parametric estimator using a regular stochastic block models and the algorithm developed in \citet{olhede2014network}. The second step smooths the step-graphon by taking the average of each block having similar density. As such, we obtain a stochastic shape model where shapes are unions of blocks. 
\begin{figure}[ht!]
    \centering
    \includegraphics[trim={0cm 8cm 0cm 0cm}, clip, width=\linewidth]{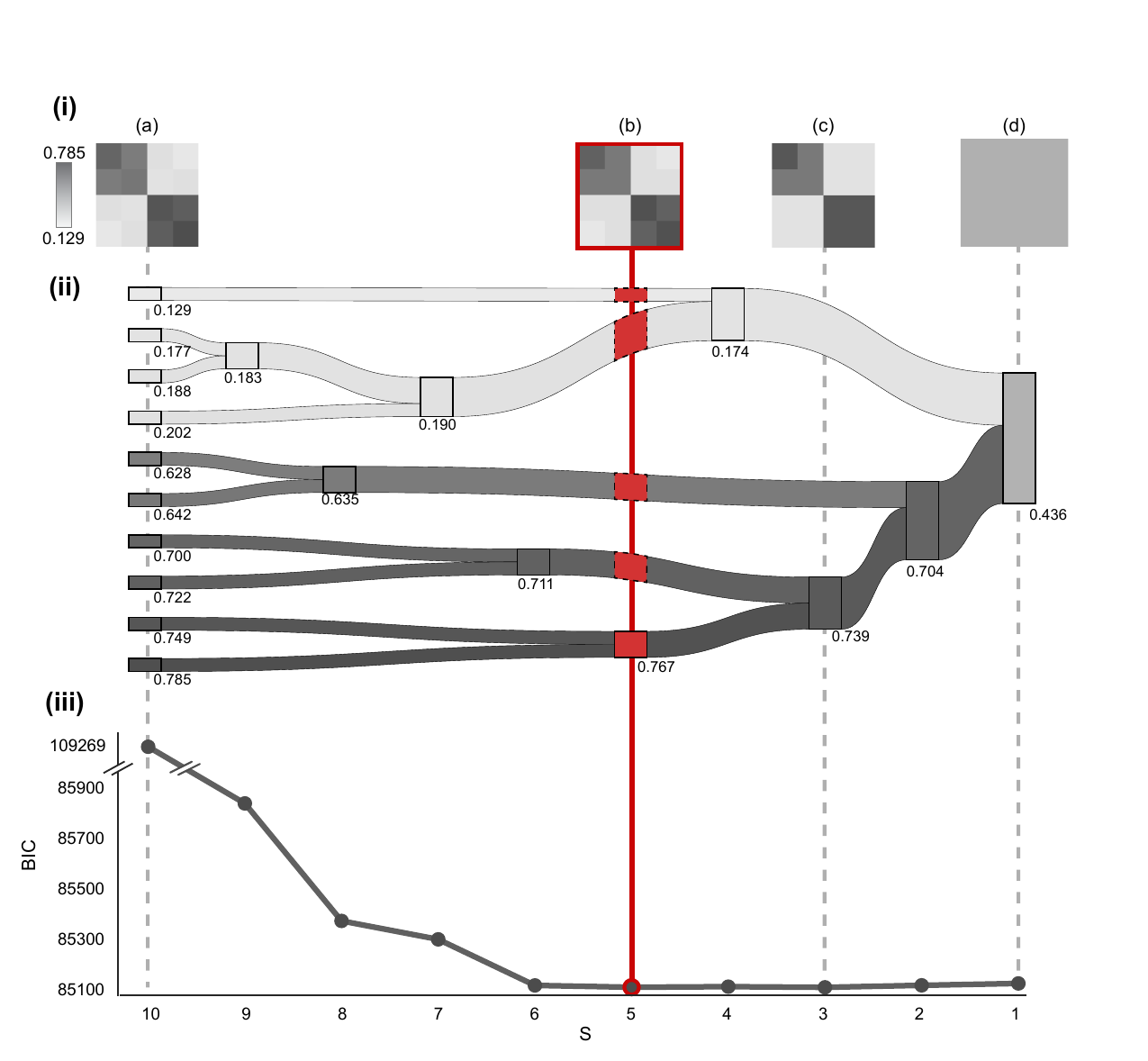}
    \caption{Comparative analysis of estimators at number of shapes and model selection based on Bayesian Information Criterion (BIC). Panel (i) presents some step-function approximations with a decreasing number of shapes from left to right (a-d). Panel (ii) is a Sankey diagram illustrating the smoothing procedure and shape's probability value, where the width of the lines indicates the number of blocks smoothed together. Panel (iii) depicts the BIC values for the corresponding number of shapes ($S$), with the lowest BIC identified at $S=5$, indicating the most efficient model according to the BIC (in red).}
    \label{fig:schema_estimator}
\end{figure}

When performing the first step of our algorithm, we obtain, say, a regular $k$-block stochastic block model, where $k$ is depending on $n$ as in \citet{olhede2014network}. Counting the intra-blocks (i.e. the diagonal ones) and the inter blocks (non-diagonal), this gives us a total of $\binom{k+1}{2}$ potential blocks to smooth. To average together blocks (tiles) of similar density, we first cluster them using \textit{k-means} algorithm, more particularly \textit{k-means++} \citep{arthur2007k} to avoid well-known initialization issues from regular \textit{k-means}. We then cluster the tiles to obtain $s$ shapes that are unions of blocks of probabilistically equivalent edge variables. We define the model obtained by smoothing into $s$ shapes as  $\mathcal{M}_{s}$ for $s\in \left[\binom{k+1}{2}\right]$. Finally, to select our final model, we compute the Bayesian Information Criterion (BIC) for each $\mathcal{M}_s$ for $s\in \left[\binom{k+1}{2}\right]$. We use BIC as an embodiment of our previous comment on Occam's razor principle, as it is a criterion known to quantitatively assess the goodness of fit of a model against its complexity. By selecting the model that minimizes BIC, we aim to get a simpler model that fit the data as well as more intricate ones. BIC is a good choice to perform such a task, as it is considered to be equivalent to the Minimum Description Length and Bayesian posterior inference in our framework of dense networks \citep{yan2014model, peixoto2021descriptive}.

\subsection{Community detection and parallel with link communities}
\label{link_com_section}
In this section, we are providing a discussion of the concept of community and how our method provides new analytical tools in community detection and network topology analysis. Much like the clustering work of \citet{hennig2015true}, we argue that defining a community based on the data alone does not reflect the idea of an unobservable underlying truth and of generalization of results to entities that were not observed. As such, one must first define what they consider to be a community before trying to find one.

A common notion of community in networks is nodal based.
As stated in \citet{gaumont2015expected}, in network analysis, a prevailing hypothesis suggests that the density of connections within node communities should surpass that of a random model without inherent community structure.
This corresponds to the basis of the notion of modularity from \citet{newmanmodularity} and assortative stochastic block model.
However, such a definition can cause issues for networks that possess dissassortative structure as well as overlapping communities. Many workarounds were proposed based on node communities, e.g. hierarchical clustering.
However, over the last decade, the notion of node communities in understanding network structure has been challenged by the notion of link communities \citep{ahn2010link, evans2009line, gaumont2015expected, he2015identification, jin2019robust}.
Partitioning links is intuitive as one could easily imagine that, in a social network, an individual can belong to multiple communities, each characterized by its interactions, its links, to these communities.
While providing a new way to analyze networks based only on the edges, they also can provide some level of description of overlapping and hierarchies of node communities, as nodes can inherit the community labels of its adjacent links \citep{ahn2010link, evans2009line}.
Link communities are particularly important in the context of weighted networks as they often better characterize community behavior and network topologies, and semantically described edges, i.e. an interaction is an email, a call \citep{jin2019robust} etc.
However, the link community scheme often generates a highly overlapping community structure even in cases where the network does not exhibit such behavior.

Methods such as \citet{he2015identification} have been proposed to get the best of both worlds using a hybrid notion of node and link communities.
As mentioned in \citet{he2015identification}, some real-world systems are better characterized by a mixture of both node and link communities. A node may belong to a node community or be connected by an edge associated with a link community, and vice versa.
While not incorporated as such in our paper, we argue that our estimator could be interpreted as a hybrid, through the invertible map between node and link communities, and provide useful information for both a node perspective and link one.
Indeed, through the correspondence between SBMs and our blocky SSM estimators, node communities comes naturally.
Link communities, however, are different as they are based on another definition than the common density argument used in \citet{ahn2010link}. In contrast, we define a link community purely based on stochastic equivalence of blocks of edges in the estimator, for which we can find similarities with the work in hierarchical community detection of \citet{schaub2023hierarchical_seep_stoch_equiv} with their stochastic diagonal/non-diagonal equivalent partition of the blocks. Note that our definition of link communities differ from methods like \citet{he2015stochastic_lc_realised_edges} or \citet{zhou2015infinite_lc_realised_edges}. Indeed, while they provide ideas closed to stochastic equivalence of edges, they only do it on realized edges. We, on the other hand, do it on edge variables, much like the concept of Mixed membership model \citep{airoldi2008mixed}.

We argue that our notion of link communities is appropriate in many applications, one of which is illustrated in the next section using the political weblogs dataset from  \citet{adamic2005political}. One could imagine the use of stochastic equivalence of interaction in sociology through notions like \textit{parallel communities}, \textit{segregated communities} or social stratification \citep{gans1982urban, putnam2015bowling}, in biology when studying \textit{convergent evolution}, in social science when studying online behaviors \citep{stoltenberg2019community}, in anthropology when studying complex systems of societal developments \citep{diamond1999guns} or even in fraud detection \citep{alexopoulos2021detecting}.

%@@@@@@@@@@@@@@@@@@@@@@@@@@@@@@@
%@@@@@@@@@@@@@@@@@@@@@@@@@@@@@@@
%@@@@@@@@@@@@@@@@@@@@@@@@@@@@@@@
\section{Experiments}
\subsection{Synthetic data}
\label{sec:synth_data}
We consider various different type of graphons to generate synthetic adjacency matrices to test the validity of our method. As benchmarks, we use the Universal Singular Value Threshold (USVT) from \citet{chatterjee2015matrix} and the Sorting-and-Smoothing algorithm from \citet{chan2014consistent}. We use the mean squared error (MSE) and area under the ROC-curve (AUC) as measures of performance in, respectively, model fit and link prediction. We selected five different models.
The first one is the latent distance model \citep{hoff2002latent}, defined as $f_0(x,y) = |x-y|$ for an example of a smooth function. 
Second and third are, respectively, assortative and dissassortative 5-blocks stochastic block models to show how robust our method can be when applied to opposite network topologies. This is because real networks may exhibit a mix of assortative and dissassortative characteristics, or may not clearly belong to one category \citep{fortunato2016community}. 
\begin{figure}[ht!]
    \centering
    \includegraphics[width = 0.8\linewidth]{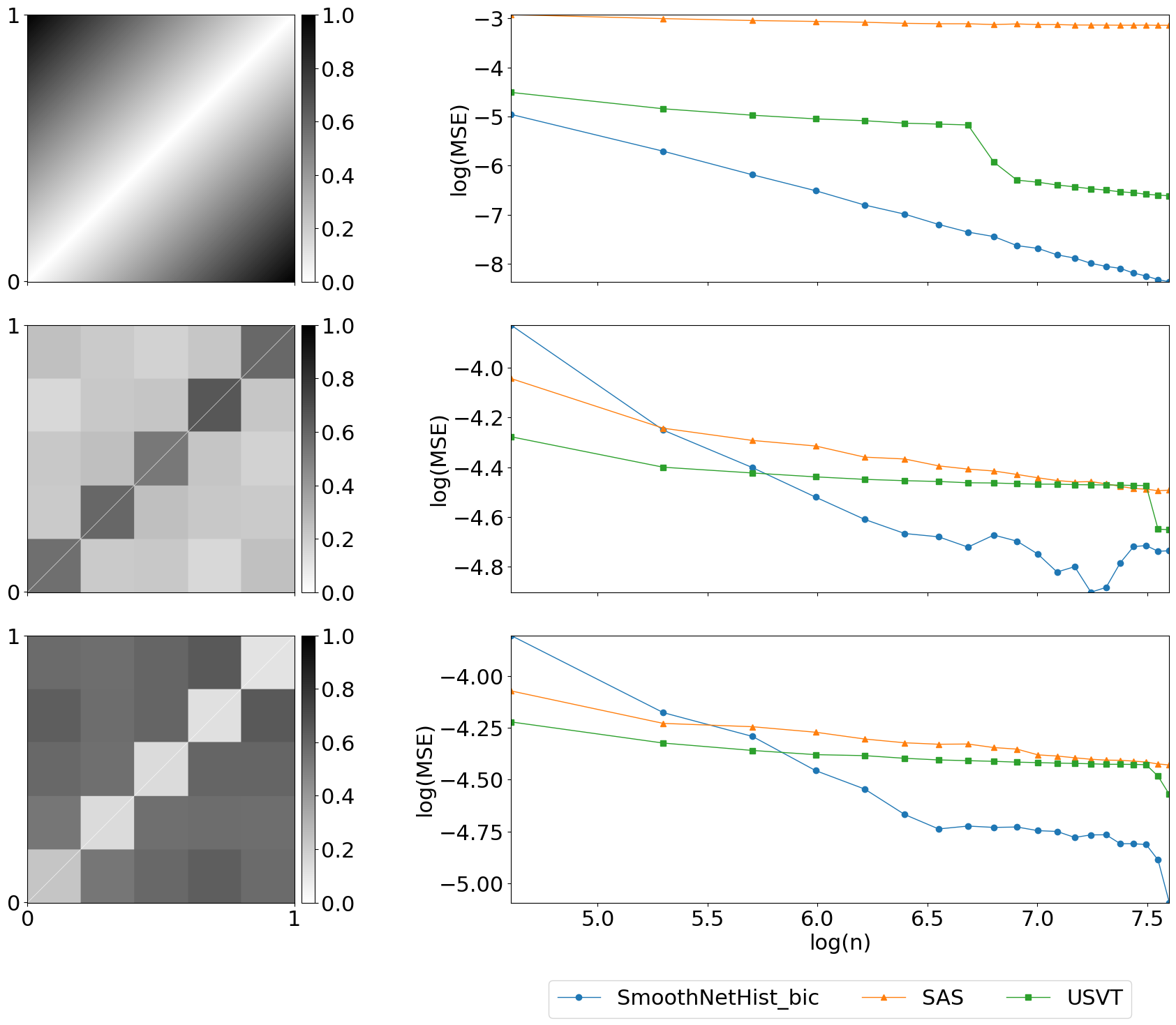}
    \caption{Performance comparison, between our method, USVT \citep{chatterjee2015matrix} and SAS \citep{chan2014consistent} over different graphon realization, linewise. Left is the functional representation of the graphon. Middle is a log-log plot of the mean squared error. Right is a plot of the area under the ROC-curve.}
    \label{fig:graphon_mse_auc}
\end{figure}

We followed the estimation procedure mentioned in the previous section. The results for each scenario are presented in Figure \ref{fig:graphon_mse_auc}. These results are obtained by averaging over 20 Monte-Carlo  simulations and for $n$ going from 100 to 2000 in steps of 100. In Figure \ref{fig:graphon_mse_auc}, we see that our method is the best candidate both in MSE and AUC for both cases of smooth functions and step functions. Thus, among the three method, for the graphon selected, our method is the best fit and the most competitive for link prediction. 
Note that the original block estimator from \citet{olhede2014network} is not present in the plot, only because the error were practically the same between the smoothed estimator and the original one. 

We illustrate this last point in Figure \ref{fig:ratiopar} by checking with different functions. The first graphon in this figure is the logit sum function $f_1(x,y) = 1/(1 + e^{-10(x^2 + y^2)})$. The second is $f_2(x,y) = \operatorname{log}(1 + 0.5\operatorname{max}(x,y))$ and the final one is a hierarchical stochastic block model with constant inter probabilities to simulate behavior from a partition model, call it $f_3$. In all three cases, we see a significant decrease in the number of parameters compared to the original estimator from \citet{olhede2014network} as $n$ increases. We also illustrate the loss in predictive power from smoothing by showing the ratio of AUC between the two methods. Here, a ratio converging to one means that our estimator is becoming as good as the original (i.e. before smoothing) for link prediction. In Figure \ref{fig:ratiopar}, we can see that smoothing is a valid approach in a bias-variance reduction approach as we have a tremendous reduction of parameters (95\% maximum reduction as seen in the case of $f_2$) for an estimator that is, from the AUC ratio, as competitive as the original estimator. 
\begin{figure}[ht!]
    \centering
    \includegraphics[width = 0.98\linewidth]{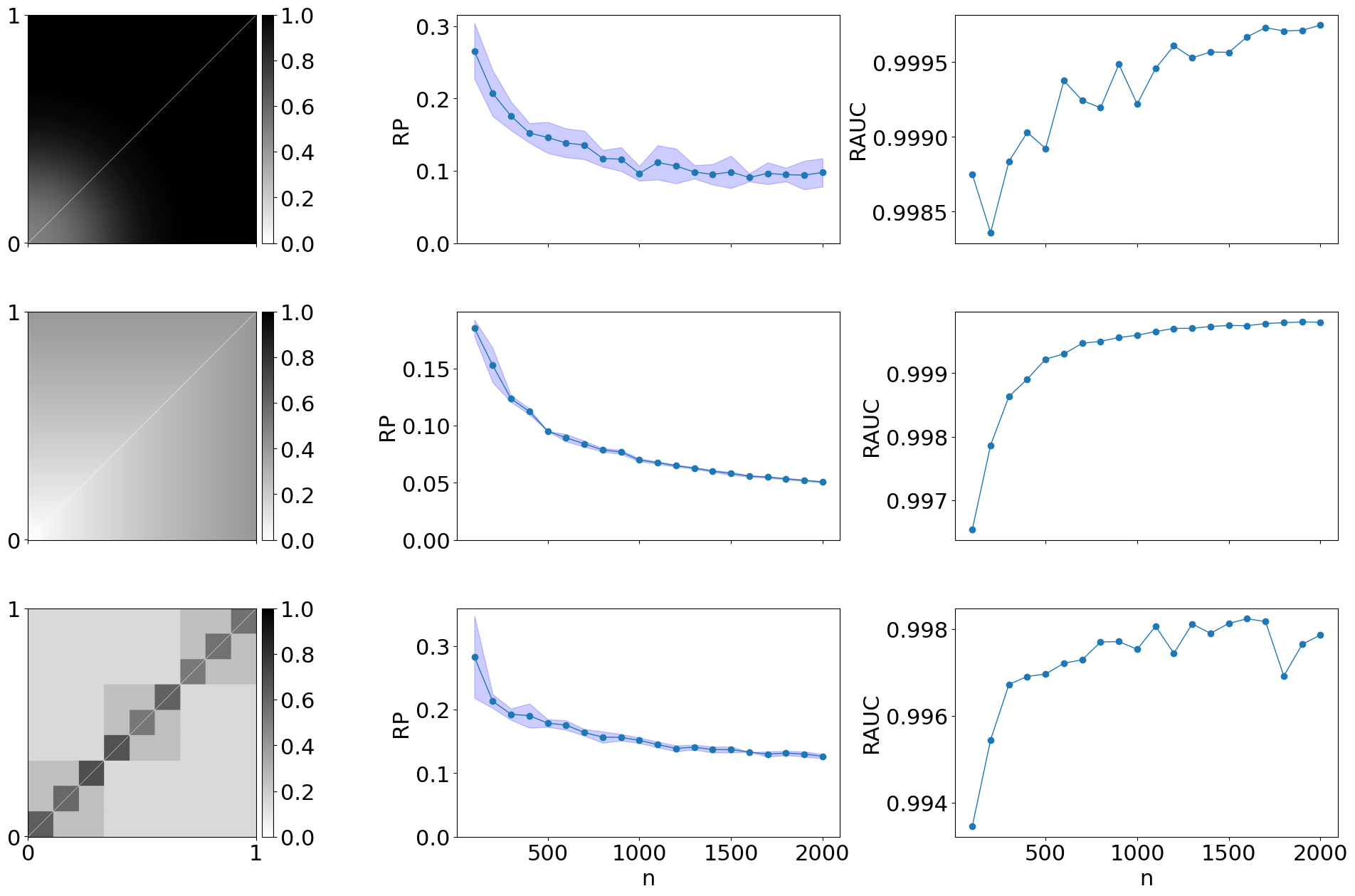}
    \caption{Plots on the left represent the graphons $f_i$ for $i = 1,2,3$ from top to bottom. Middle plots represent the ratio of the number of parameters (RP) between our estimator, using BIC, and the estimator from \citet{olhede2014network}. Similarly, plots on the right represent the AUC ratio (RAUC) between the two methods. All averaged over 10 Monte Carlo simulations.}
    \label{fig:ratiopar}
\end{figure}
%@@@@@@@@@@@@@@@@@@@@@@@@@@@@@@@
%@@@@@@@@@@@@@@@@@@@@@@@@@@@@@@@
%@@@@@@@@@@@@@@@@@@@@@@@@@@@@@@@
\subsection{Real-world networks}
\label{polblog_section}
To further illustrate the performance of our estimator and its versatility as a tool for exploratory data analysis, we selected real-world datasets such as the political weblogs \citep{adamic2005political}. This network describes an interaction between political blogs $i$ and $j$ if a hyperlink was directing blog $i$ to blog $j$ and/or vice versa. 
The network itself is composed of $N=1224$ nodes for $E=16783$ edges when removing isolated nodes.
A key feature of this dataset is the categorization of blogs based on their political alignment – typically as conservative or liberal. This division explains why this dataset is used as an example for community detection.
As suggested in \citet{peixoto2014hierarchical}, the highest-scale division is this bimodal partition between liberals and conservatives, both in hierarchical and non-hierarchical block model approaches.
However, by considering a finer partition of this network, say 15 to 20 node communities, one can characterize the heterogeneous patterns of interactions within the network. Thus, one can find different topologies and structures within a network based on how you order the nodes, as illustrated in Figure \ref{fig:polblog_adj}.
\begin{figure}[ht!]
    \centering
    \includegraphics[width = \linewidth]{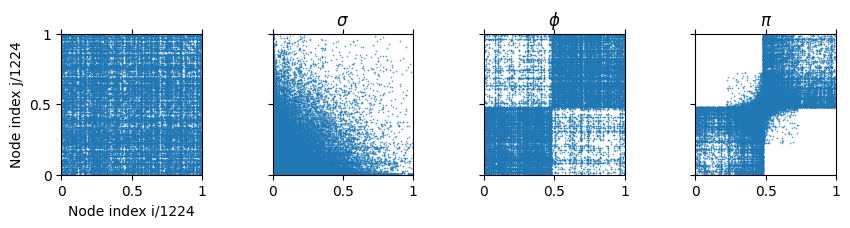}
    \caption{Adjacency matrix of the political weblogs dataset \citet{adamic2005political} when permuting the node labels according, from left to right, to the identity map (no permutation), to degrees ($\sigma$ permutation), to political party ($\phi$ permutation) and degrees based on interactions to the opposite party ($\pi$ permutation).}
    \label{fig:polblog_adj}
\end{figure}
As one can see in Figure \ref{fig:polblog_est}, both political groups have subgroups that are cited across the political board, acting as bridges. While we can also see other pattern between groups, like citation patterns across subgroups.
A pattern that stands out when we consider the notion of link communities based on stochastic equivalence, as in section \ref{link_com_section}, is the stochastic shape with the lowest density on the right of Figure 6. Indeed, we see that this shape in particular is a similar fit to the low-density interaction pattern across the segregated political board, i.e. the $\pi$ permutation. As a tool for community-based analysis of networks, this unifies barely-heterogeneous interactions between groups into one component. This component could be seen, for example, as the block of politically segregated groups' link community, essentially easing analysis by reducing the number of parameters of the stochastic block estimator as is the goal of our method.
\begin{figure}
    \centering
    \includegraphics[width = 0.5\linewidth]{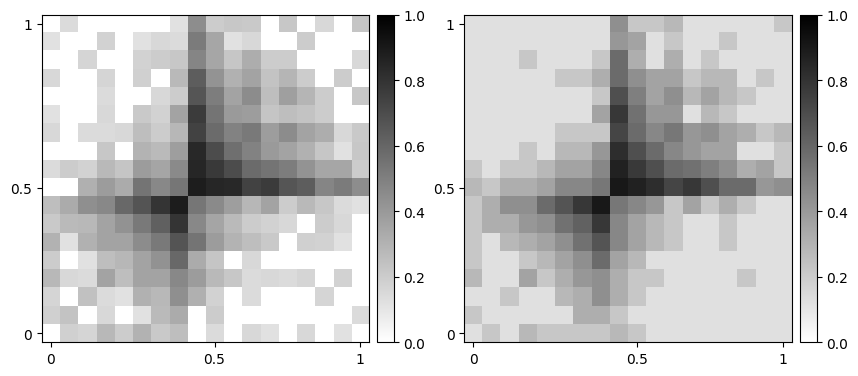}
    \caption{Network histogram estimators $\hat{f}(x,y)^{1/4}$. Left is the estimator from \citet{olhede2014network}, right is our estimator. The fourth root is used for ease of visualisation.}
    \label{fig:polblog_est}
\end{figure}

\section{Conclusion}
\label{sec:conc}

Here, we have provided a new latent class model, that we call the \textit{stochastic shape model}, for unlabeled networks with its own estimation method inspired from \citet{olhede2014network}.
The approach reduces variance introduction to the estimator's error at the price of bias while still maintaining similar, if not better, empirical predictive performances compared to methods based on stochastic block models, as seen in Section 5.
Consequently, we obtain a simpler representation of our estimator which is more fitted as a nonparametric statistical summary of network interactions.
Indeed, while method like \citet{olhede2014network} or \citet{gao2015rate} improve their approximation of the underlying data-generating mechanism by introducing more blocks -- therefore increasing complexity -- we, here, control this added complexity by a smoothing operation based on similarity of block density.
To check that smoothing does not negatively impact inference, we show that the rate of convergence of our method is optimal when the generating mechanism of the network is a stochastic shape model and Hölder-smooth function.
As a sanity check, we also provide experimental results that illustrate the predictive performance of our method.
While providing an inherent node analysis/community detection thanks to its block structure, we argue that our method also provide a description of communities of interactions based on stochastic equivalence. Thus providing a new potential definition for link communities.
Moreover, the sequential smoothing procedure can be further tuned to obtain a multiscale estimator, providing different levels of link communities size.

\bigskip
\begin{center}
{\large\bf SUPPLEMENTARY MATERIAL}
\end{center}

\begin{description}

\item[Acknowledgements:] This work was supported by the European Research Council under Grant CoG 2015-682172NETS, within the Seventh European Union Framework Program.

\item[Title:] "Supplementary material" contains the proofs and further auxiliary results needed to obtain the theoretical results presented in this paper.

\item[Code:] Method is implemented and is available in \citet{dufour_2023}.

\item[Political weblogs data set:] Data set used in the section \ref{polblog_section} was extracted from \url{https://networks.skewed.de/net/polblogs}.

\end{description} 

\appendix

\section{Auxiliary results and oracle inequalities}

We will now list results and proofs that are essential for a good understanding of the proof in our paper.

\subsection{Auxiliary results}
\begin{itemize}
    \item \textit{Hoeffding's inequality: }We recall it from \cite{vershynin2010introduction} as this is essential for non-parametric rates.
    \begin{Proposition}[Bernstein-type inequality]
    \label{bernstein}
    Let $X_{1}, \ldots, X_{N}$ be independent centred sub-exponential random variables, and $K=\max _{i}\left\|X_{i}\right\|_{\psi_{1}}$. Then for every $a=$ $\left(a_{1}, \ldots, a_{N}\right) \in \mathbf{R}^{N}$ and every $t \geq 0$, we have
    $$
    \mathbb{P}\left\{\left|\sum_{i=1}^{N} a_{i} X_{i}\right| \geq t\right\} \leq 2 \exp \left[-c \min \left(\frac{t^{2}}{K^{2}\|a\|_{2}^{2}}, \frac{t}{K\|a\|_{\infty}}\right)\right],
    $$
    where $e>0$ is an absolute constant. 
    \end{Proposition}
    \item \textit{Least-square minimiser estimator:} This is the proof of the main motivational result from this paper, Proposition 3.1.
    \begin{proof}[Proof of Proposition 3.1]
    \label{proofminLoss}
    Under the model specification of the Stochastic Shape Model (Definition 2.1), we have that the block averages $\bar{A}_{ab}$, for $a,b\in[k]$, are similar over areas $S_c$ for $c\in[s]$. Defining $h_c:=|S_c|h^2$, the least square, for $1\leq i<j\leq n$, is thus
    $$
    \begin{aligned}
        \mathcal{L}(A;u\circ z^2, \theta)&=\sum_{c=1}^s\sum_{(i,j)\in w^{-1}(c)}\left\{A_{ij}\operatorname{log}\left(\theta_{w(i,j)}\right)+\left(1-A_{ij}\right)\operatorname{log}\left(1-\theta_{w(i,j)}\right)\right\}\\
        &= \sum_{c=1}^s \sum_{\substack{a \in u^{-1}(c) \\ b \in u^{-1}(c)}}\sum_{\substack{i \in z^{-1}(a) \\ j \in z^{-1}(b)}} \left\{A_{ij}\operatorname{log}\left(\theta_{u(z(i),z(j))}\right)+\left(1-A_{ij}\right)\operatorname{log}\left(1-\theta_{u(z(i),z(j))}\right)\right\}\\
        &=\sum_{c=1}^s \sum_{\substack{a \in u^{-1}(c) \\ b \in u^{-1}(c)}}\operatorname{log}\left(\theta_{u(a,b)}\right)\sum_{\substack{i \in z^{-1}(a) \\ j \in z^{-1}(b)}}A_{ij}+\sum_{c=1}^s \sum_{\substack{a \in u^{-1}(c) \\ b \in u^{-1}(c)}}\operatorname{log}\left(1-\theta_{u(a,b)}\right)\sum_{\substack{i \in z^{-1}(a) \\ j \in z^{-1}(b)}}\left(1-A_{ij}\right)\\
        &=\sum_{c=1}^s\operatorname{log}\left(\theta_{c}\right) \sum_{\substack{a \in u^{-1}(c) \\ b \in u^{-1}(c)}}\sum_{\substack{i \in z^{-1}(a) \\ j \in z^{-1}(b)}}A_{ij}+\sum_{c=1}^s\operatorname{log}\left(1-\theta_{c}\right) \sum_{\substack{a \in u^{-1}(c) \\ b \in u^{-1}(c)}}\sum_{\substack{i \in z^{-1}(a) \\ j \in z^{-1}(b)}}\left(1-A_{ij}\right)\\
        &=\sum_{c=1}^s\operatorname{log}\left(\theta_{c}\right)\sum_{(i,j)\in w^{-1}(c)}A_{ij}+\sum_{c=1}^s\operatorname{log}\left(1-\theta_{c}\right)\sum_{(i,j)\in w^{-1}(c)}\left(1-A_{ij}\right)\\
        &= \sum_{c=1}^s h_c\left\{\operatorname{log}\left(\theta_{c}\right)\Bar{\Bar{A}}_c + \operatorname{log}\left(1-\theta_{c}\right)\left(1-\Bar{\Bar{A}}_c\right)\right\},
    \end{aligned}
    $$
    which is maximised for $\theta_c= \Bar{\Bar{A}}_c$. Profiling $\theta$ out of the least square gives as in \citet{bickel2009nonparametric}
    $$
    \begin{aligned}
    \mathcal{L}(A ; u\circ z^2) & =\max _{\theta \in[0,1]^{s}} \mathcal{L}(A ; u\circ z^2, \theta) \\
    &=\sum_{c=1}^s h_c\left\{\operatorname{log}\left(\theta_{c}\right)\Bar{\Bar{A}}_c + \operatorname{log}\left(1-\theta_{c}\right)\left(1-\Bar{\Bar{A}}_c\right)\right\}\\
    &=\sum_{i<j}\left\{A_{i j} \log \Bar{\Bar{A}}_{w(i,j)}+\left(1-A_{i j}\right) \log \left(1-\Bar{\Bar{A}}_{w(i,j)}\right)\right\}.
    \end{aligned}
    $$
    
    As stated in \citet{wolfe2013nonparametric}, any maximiser of this last line is a maximum profile likelihood estimator (MPLE) for $z$ in $\mathcal{Z}_{n,k}$. Applying their reasoning, it further follows that maximising such a likelihood is equivalent to minimising the sum of the Bernoulli Kullback-Leibler divergences $\sum_{i<j} \mathrm{D}\left(A_{i j} \| \Bar{\Bar{A}}_{w(i,j)}\right)$. The rest follows by using Lemma C.9 from \citet{wolfe2013nonparametric}.
    \end{proof}
    \item \textit{A bound on using packing number: }A key contribution from \citet{gao2015rate} over \citet{wolfe2013nonparametric} is the use of covering and packing number to obtain rate optimality. The following result helps us introduce this notion in the coming proofs. Note that this is a direct adaptation from Lemma A.1 in \citet{gao2015rate} modified to fit the stochastic shape framework.
    
    \begin{Lemma}[\citep{gao2015rate}]
    \label{lemmaA.1}
    Let $\mathcal{B} \subset\left\{a \in \mathbb{R}^{n \times n}: \sum_{i j} a_{i j}^2 \leq 1\right\}$. Assume for any $a, b \in$ $\mathcal{B}$,
    \begin{equation}
    \label{NormCondLemma}
    \frac{a-b}{\|a-b\|} \in \mathcal{B}.
    \end{equation}
    Then, we have
    $$
    \mathbb{P}\left(\sup _{a \in \mathcal{B}}\left|\sum_{i j} a_{i j}\left(A_{i j}-\theta_{i j}\right)\right|>t\right) \leq \mathcal{N}(1 / 2, \mathcal{B},\|\cdot\|) \exp \left(-C t^2\right),
    $$
    for some universal constant $C>0$.
    \end{Lemma}
    
    \begin{proof}[Proof of Lemma \ref{lemmaA.1}]
    Let $\mathcal{B}^{\prime}$ be a $1 / 2$-net of $\mathcal{B}$ such that $\left|\mathcal{B}^{\prime}\right| \leq \mathcal{N}(1 / 2, \mathcal{B}
    \|\cdot\|)$ and for any $a \in \mathcal{B}$, there is $b \in \mathcal{B}^{\prime}$ satisfying
    \begin{equation}
    \label{boundHalfLemma}
    \| a-b|| \leq 1 / 2 .
    \end{equation}
    Thus,
    $$
    \begin{aligned}
    |\langle a, A-\theta\rangle| & \leq|\langle a-b, A-\theta\rangle|+|\langle b, A-\theta\rangle| \\
    & \leq\|a-b\|\left|\left\langle\frac{a-b}{\|a-b\|}, A-\theta\right\rangle\right|+|\langle b, A-\theta\rangle| \\
    & \leq \frac{1}{2} \sup _{a \in \mathcal{B}}|\langle a, A-\theta\rangle|+|\langle b, A-\theta\rangle|,
    \end{aligned}
    $$
    where the last inequality is due to (\ref{boundHalfLemma}) and the assumption (\ref{NormCondLemma}). Taking sup with respect to $\mathcal{B}$ and max with respect to $\mathcal{B}'$ on both sides, we have
    $$
    \sup _{a \subset \mathcal{B}}\left|\sum_{ij} a_{ij}\left(A_{ij}-\theta_{ij}\right)\right| \leq 2 \max _{b \in \mathcal{B}^{\prime}}\left|\sum_{ij} b_{ij}\left(A_{ij}-\theta_{ij}\right)\right| .
    $$
    
    From this inequality follows by linearity the cumulative function
    $$
    \mathbb{P}\left(\sup _{a \subset \mathcal{B}}\left|\sum_{ij} a_{ij}\left(A_{ij}-\theta_{ij}\right)\right|\geq t\right) \leq \mathbb{P}\left(2 \max _{b \in \mathcal{B}^{\prime}}\left|\sum_{ij} b_{ij}\left(A_{ij}-\theta_{ij}\right)\right|\geq t\right).
    $$
    By the union bound we then get
    \begin{align*}
    \mathbb{P}\left( \max_{b \in \mathcal{B}^{\prime}}\left|\sum_{ij} b_{ij}\left(A_{ij}-\theta_{ij}\right)\right|\geq \frac{t}{2}\right) &\leq \sum_{b\in\mathcal{B}'}\mathbb{P}\left(\left|\sum_{ij} b_{ij}\left(A_{ij}-\theta_{ij}\right)\right|\geq \frac{t}{2}\right)\\
    &\stackrel{(*)}{\leq} \sum_{b\in\mathcal{B}'}\operatorname{exp}\left(\frac{-\frac{t^2}{2}}{\|b\|_2^2}\right)\\
    &\stackrel{(**)}{\leq} \sum_{b\in\mathcal{B}'}\operatorname{exp}\left(-\frac{t^2}{2}\right)\\
    &= \left|\mathcal{B}'\right|\operatorname{exp}\left(-\frac{t^2}{2}\right)\\
    &\leq \mathcal{N}(1 / 2, \mathcal{B},\|\cdot\|) \exp \left(-C t^{2}\right),
    \end{align*}
    where (*) follows from Hoeffding's inequality \citep{vershynin2010introduction}[Prop 5.16] and (**) follows from $\mathcal{B} \subset\left\{a \in \mathbb{R}^{s}: \sum_{i} a_{i}^{2} \leq 1\right\}$. Thus, the proof is complete.
    \end{proof}

    \item \textit{Cardinality of the operator $w$: }As seen in the previous lemma, for oracle inequalities, we need a bound on the cardinality of the operator $w\in\mathcal{W}_{n,s}$. We are now going to prove that $\left|\mathcal{W}_{n,s}\right|\leq \operatorname{max}(k,s)^{n}$.
    
    \begin{proof}
        First note that, as the mapping $z$ and $u$ are surjective mappings, the cardinalities of their respective operator set are bounded, in order by $k^{n-(k-1)}$ and $s^{k^2-(s-1)}/k!$. Using these as well as assuming that $k^2-s<n$, we have
        Using the fact that $k\leq s \leq k^2$, the above can be bounded by 
        \begin{align*}
            \left|\mathcal{W}_{n,s}\right|&\leq \operatorname{exp}\left((n-k+1)\operatorname{log}(k)+(k^2-2+1)\operatorname{log}(s)- \operatorname{log}(k!)\right)\\
            &\stackrel{(*)}{\leq} \operatorname{exp}\left((n-k+1+k^2-s-k)\operatorname{log}(\operatorname{max}(k,s))\right)\\
            &\stackrel{(**)}{\leq} \operatorname{exp}\left(n\operatorname{log}(\operatorname{max}(k,s))\right),
        \end{align*}
        where (*) follows by Stirling's approximation and (**) follows from the assumption $n>\operatorname{max}(0,k^2-s)$.
    \end{proof}
\end{itemize}

\subsection{Oracle inequalities}
%---- Oracle Inequalities ----
Following a similar argument as in \citet{gao2015rate}, we shall obtain oracle inequalities to further derive the rate of convergence of our estimator. We denote the true value on each shape by $\left\{Q_{c}^{*}\right\} \in[0,1]^{s}$ and the oracle assignment by $w^{*} \in \mathcal{W}_{n, s}$ such that $\theta_{ij}=Q_{w^{*}(i,j)}^{*}$ for any $(i,j)\in [n]^2$. To facilitate the proof, we introduce the following notation. For the estimated $\hat{w}$, define $\left\{\hat{Q}_{c}\right\} \in[0,1]^{s}$ by $\tilde{Q}_{c}=\bar{\theta}_{c}(\hat{w}) $, and also define $\tilde{\theta}_{ij}=\tilde{Q}_{\hat{w}(i,j)}$ for any $(i,j)\in [n]^2$. Recall for any optimiser of the objective function, $(\hat{Q}, \hat{w}) \in \underset{Q \in \mathbb{R}^{s}, w \in \mathcal{W}_{n, s}}{\operatorname{argmin}} L(Q, w),$. By the definition of this estimator, we have $L(\hat{Q}, \hat{w}) \leq L\left(Q^{*}, w^{*}\right)$ which, as $\|\hat{\theta}-A\|^{2}=\|\hat{\theta}-\theta\|^{2}-2\langle\hat{\theta}-\theta, A-\theta\rangle+\left\|\theta-A\right\|^{2}$, can be rewritten as 
\begin{equation}
\label{4.1paper}
\|\hat{\theta}-A\|^{2} \leq\|\theta-A\|^{2}\quad \iff \quad  \|\hat{\theta}-\theta\|^{2} \leq 2\langle\hat{\theta}-\theta, A-\theta\rangle.
\end{equation}
As in \citet{gao2015rate}, the right-hand side of (\ref{4.1paper}) can be bounded as
\begin{align*}
\langle\hat{\theta}-\theta, A-\theta\rangle\leq &\|\hat{\theta}-\tilde{\theta}\|\left|\left\langle\frac{\hat{\theta}-\tilde{\theta}}{\|\hat{\theta}-\tilde{\theta}\|}, A-\theta\right\rangle\right|+\left(\|\tilde{\theta}-\hat{\theta}\|+\|\hat{\theta}-\theta\|\right)\left|\left\langle\frac{\tilde{\theta}-\theta}{\|\tilde{\theta}-\theta\|}, A-\theta\right\rangle\right|.
\end{align*}
Thus, we need to bound the following three terms
\begin{equation*}
\|\hat{\theta}-\tilde{\theta}\|, \quad\left|\left\langle\frac{\hat{\theta}-\tilde{\theta}}{\|\hat{\theta}-\tilde{\theta}\|}, A-\theta\right\rangle\right|, \quad\left|\left\langle\frac{\tilde{\theta}-\theta}{\|\tilde{\theta}-\theta\|}, A-\theta\right\rangle\right|,
\end{equation*}
to obtain our error bound and determine the rate of convergence of such a method. To do so, we need the following lemmas that are adapted from \citet{gao2015rate} to the stochastic shape setting.
%@@@@@@@@@@@@@@@@@@@@@@@@@@@@@@@@@@@@@@@@@@@@@@@@
\begin{Lemma}
\label{Lemma4.3}
For any constant $C^{\prime}>0$, there exists a constant $C>0$ only depending on $C^{\prime}$, such that
$$
\left|\left\langle\frac{\hat{\theta}-\tilde{\theta}}{\|\hat{\theta}-\tilde{\theta}\|}, A-\theta\right\rangle\right| \leq C \sqrt{s+n \operatorname{log}(\operatorname{max}(k,s))},
$$
with probability at least $1-\exp \left(-C^{\prime} n \operatorname{log}(\operatorname{max}(k,s))\right)$.
\end{Lemma} 
\begin{proof}
For each $w \in \mathcal{W}_{n, s}$, define the set $\mathcal{B}_{w}$ by $\mathcal{B}_{w}=\left\{\left\{c_{ij}\right\}: c_{ij}=Q_{c}\right.$ if $(i,j) \in w^{-1}(c)$ for some $Q_{c}$, and $\left.\sum_{ij} c_{ij}^{2} \leq 1\right\}$.
In other words, $\mathcal{B}_{w}$ collects those union of piecewise constant matrices determined by $w$. Thus, we have the bound
$$
\left|\sum_{ij} \frac{\tilde{\theta}_{ij}-\hat{\theta}_{ij}}{\sqrt{\sum_{ij}\left(\tilde{\theta}_{ij}-\hat{\theta}_{ij}\right)^{2}}}\left(A_{ij}-\theta_{ij}\right)\right| \leq \max _{w \in \mathcal{W}_{n, s}} \sup _{c \in \mathcal{B}_{w}}\left|\sum_{ij} c_{ij}\left(A_{ij}-\theta_{ij}\right)\right| .
$$
Note that for each $w \in \mathcal{W}_{n, s}, \mathcal{B}_{w}$ satisfies the condition (\ref{NormCondLemma}). Thus, we have
$$
\begin{aligned}
& \mathbb{P}\left(\max _{w \in \mathcal{W}_{n, s}} \sup _{c \in \mathcal{B}_{w}}\left|\sum_{ij} c_{ij}\left(A_{ij}-\theta_{ij}\right)\right|>t\right) \\
\leq & \sum_{w \in \mathcal{W}_{n, s}} \mathbb{P}\left(\sup _{c \in \mathcal{B}_{w}}\left|\sum_{ij} c_{ij}\left(A_{ij}-\theta_{ij}\right)\right|>t\right) \\
\leq & \sum_{w \in \mathcal{W}_{n, s}} \mathcal{N}\left(1 / 2, \mathcal{B}_{w},\|\cdot\|\right) \exp \left(-C_{1} t^{2}\right),
\end{aligned}
$$
for some universal $C_{1}>0$, where the last inequality is due to lemma \ref{lemmaA.1}.

\begin{Remark}
From \citet{pollard1990empirical}, Lemma 4.1, we have the following result on covering and packing numbers
\begin{equation}
    \label{boundCovering}
    \mathcal{N}\left(\epsilon, \mathcal{B},\|\cdot\|\right) \leq \mathcal{D}\left(\epsilon, \mathcal{B},\|\cdot\|\right) \leq \left(\frac{3 R}{\epsilon}\right)^{V},
\end{equation}
where $\mathcal{B}$ is a subset of a $V$-dimensional affine subspace of $\mathbb{R}^n$ of diameter $R$. 
\end{Remark}

Following the previous remark, since $\mathcal{B}_{w}$ has a degree of freedom $s$, we have $\mathcal{N}\left(1 / 2, \mathcal{B}_{w},\|\cdot\|\right) \leq \left(\frac{3}{2}\right)^{s} \leq \exp \left(C_2 s\right)$ for all $w \in \mathcal{W}_{n, s}$ where $C_2 = \operatorname{log}\left(\frac{3}{2}\right)$, by \eqref{boundCovering}. Finally, by $\left|\mathcal{W}_{n, s}\right| \leq \exp (n \operatorname{log}(\operatorname{max}(k,s)))$, we have
$$
\mathbb{P}\left(\max _{w\in \mathcal{W}_{n, s}} \sup _{c \in \mathcal{B}_{w}} \sum_{i \neq j} c_{ij}\left(A_{ij}-\theta_{ij}\right)>t\right) \leq \exp \left(-C_{1} t^{2}+C_{2} s+n \operatorname{log}(\operatorname{max}(k,s))\right) .
$$
Choosing $t^{2} \propto s+n \operatorname{log}(\operatorname{max}(k,s))$, the proof is complete.
\end{proof}

\begin{Lemma}
\label{Lemma4.2}
For any constant $C^{\prime}>0$, there exists a constant $C>0$ only depending on $C^{\prime}$, such that
$$
\left|\left\langle\frac{\tilde{\theta}-\theta}{\|\tilde{\theta}-\theta\|}, A-\theta\right\rangle\right| \leq C \sqrt{n\operatorname{log}(\operatorname{max}(k,s))},$$
with probability at least $1-\exp \left(-C^{\prime} n \operatorname{log}(\operatorname{max}(k,s))\right)$.
\end{Lemma}

\begin{proof}
Note that
$$
\tilde{\theta}_{ij}-\theta_{ij}=\sum_{c \in[s]} \bar{\theta}_{c}(\hat{w}) \mathbb{I}\left\{(i,j) \in \hat{w}^{-1}(c)\right\}-\theta_{ij}
$$
is a function of the partition $\hat{w}$, then we have
$$
\left|\sum_{ij} \frac{\tilde{\theta}_{ij}-\theta_{ij}}{\sqrt{\sum_{ij}\left(\tilde{\theta}_{ij}-\theta_{ij}\right)^{2}}}\left(A_{ij}-\theta_{ij}\right)\right| \leq \max _{w \in \mathcal{W}_{n, s}}\left|\sum_{ij} \gamma_{ij}(w)\left(A_{ij}-\theta_{ij}\right)\right|,
$$
where
$$
\gamma_{ij}(z) \propto \sum_{c \in[s]} \bar{\theta}_{c}(w) \mathbb{I}\left\{(i,j) \in w^{-1}(c)\right\}-\theta_{ij}
$$
satisfies $\sum_{ij} \gamma_{ij}(w)^{2}=1$. By Hoeffding's inequality \citep{vershynin2010introduction}[Prop 5.10] and union bound (as in the proof of the Lemma (\ref{lemmaA.1})), we have
$$
\begin{aligned}
& \mathbb{P}\left(\max _{w \in \mathcal{W}_{n, s}}\left|\sum_{ij} \gamma_{ij}(w)\left(A_{ij}-\theta_{ij}\right)\right|>t\right) \\
\leq & \sum_{w \in \mathcal{W}_{n, s}} \mathbb{P}\left(\left|\sum_{ij} \gamma_{ij}(w)\left(A_{ij}-\theta_{ij}\right)\right|>t\right) \\
\leq &\left|\mathcal{W}_{n, s}\right| \exp \left(-C_{1} t^{2}\right) \\
\leq & \exp \left(-C_{1} t^{2}+n \operatorname{log}(\operatorname{max}(k,s))\right),
\end{aligned}
$$
for some universal constant $C_{1}>1$. Choosing $t^2 \propto n \operatorname{log}(\operatorname{max}(k,s))$, we have
$$
\begin{aligned}
1-\mathbb{P}\left(\max _{w \in \mathcal{W}_{n, s}}\left|\sum_{ij} \gamma_{ij}(w)\left(A_{ij}-\theta_{ij}\right)\right|\leq \sqrt{n \operatorname{log}(\operatorname{max}(k,s))}\right) &\leq \exp \left(-(C_{1}-1) n\operatorname{log}(\operatorname{max}(k,s))\right)\\
\mathbb{P}\left(\max _{w \in \mathcal{W}_{n, s}}\left|\sum_{ij} \gamma_{ij}(w)\left(A_{ij}-\theta_{ij}\right)\right|\leq \sqrt{n \operatorname{log}(\operatorname{max}(k,s))}\right) &\geq 1-\exp \left(-(C_{1}-1) n\operatorname{log}(\operatorname{max}(k,s))\right).
\end{aligned}
$$
Thus, the proof is complete.
\end{proof}

\begin{Lemma}
\label{Lemma4.1}
For any constant $C^{\prime}>0$, there exists a constant $C>0$ only depending on $C^{\prime}$, such that
$$
\|\hat{\theta}-\tilde{\theta}\| \leq C \sqrt{s+n \operatorname{log}(\operatorname{max}(k,s))},
$$
with probability at least $1-\exp \left(-C^{\prime} n \operatorname{log}(\operatorname{max}(k,s))\right)$.
\end{Lemma}

\begin{proof}
By the definitions of $\hat{\theta}_{ij}$ and $\tilde{\theta}_{ij}$, we have
$$
\hat{\theta}_{ij}-\tilde{\theta}_{ij}=\hat{Q}_{\hat{w}(i,j)}-\tilde{Q}_{\hat{w}(i,j) }=\bar{\bar{A}}_{c}(\hat{w})-\bar{\theta}_{c}(\hat{w}),
$$
for any $(i,j) \in \hat{w}^{-1}(c)$. Then
\begin{align}
\sum_{ij}\left(\hat{\theta}_{ij}-\tilde{\theta}_{ij}\right)^{2} & \leq \sum_{c \in[s]}\left|\hat{w}^{-1}(c)\right|\left(\bar{\bar{A}}_{c}(\hat{w})-\bar{\theta}_{c}(\hat{w})\right)^{2} \nonumber\\
& \leq \max _{w \in \mathcal{W}_{n, s}} \sum_{c \in[s]}\left|w^{-1}(c)\right|\left(\bar{\bar{A}}_{c}(w)-\bar{\theta}_{c}(w)\right)^{2}.\label{A.8}
\end{align}
For any $c \in[s]$ and $w \in \mathcal{W}_{n, s}$, define $n_{c}=\left|w^{-1}(c)\right|$ and $V_{c}(w)=n_{c}\left(\bar{\bar{A}}_{c}(w)-\right.$ $\left.\bar{\theta}_{c}(w)\right)^{2}$. Then, (\ref{A.8}) is bounded by
\begin{equation}
\label{4.1BoundV}
\max _{w \in \mathcal{W}_{n, s}} \sum_{c\in[s]} \mathbb{E}\left[ V_{c}(w)\right]+\max _{w \in \mathcal{W}_{n, s}} \sum_{c\in[s]}\left(V_{c}(w)-\mathbb{E}\left[ V_{c}(w)\right]\right) .
\end{equation}
We are going to bound the two terms separately. For the first term, we have
$$
\begin{aligned}
\mathbb{E} \left[V_{c}(w)\right] &=n_{c} \mathbb{E}\left[\left(\frac{1}{n_{c} } \sum_{(i,j) \in w^{-1}(c)}\left(A_{ij}-\theta_{ij}\right)\right)^2\right] \\
&=\frac{1}{n_{c}} \sum_{(i,j) \in w^{-1}(c)} \operatorname{Var}\left(A_{ij}\right) \leq 1,
\end{aligned}
$$
where we have used the fact that $\mathbb{E}\left[ A_{ij}\right]=\theta_{ij}$ and $\operatorname{Var}\left(A_{ij}\right)=\theta_{ij}\left(1-\theta_{ij}\right) \leq 1$. Summing over $c \in[s]$, we get
\begin{equation}
\label{4.1BoundV2}
\max _{w \in \mathcal{W}_{n, s}} \sum_{c \in[s]} \mathbb{E}\left[ V_{c}(w)\right] \leq C_{1} s,
\end{equation}
for some universal constant $C_{1}>0$. By Hoeffding inequality \citep{hoeffding1994probability} and $0 \leq$ $A_{ij} \leq 1$, for any $t>0$ we have
$$
\mathbb{P}\left(V_{c}(w)>t\right)=\mathbb{P}\left(\left|\frac{1}{n_{c}} \sum_{(i,j) \in w^{-1}(c)}\left(A_{ij}-\theta_{ij}\right)\right|>\sqrt{\frac{t}{n_{c}}}\right) \leq 2 \exp (-2 t).
$$
Thus, $V_{c}(w)$ is a sub-exponential random variable. By Proposition \ref{bernstein}, i.e. Bernstein's inequality for sub-exponential variables [\citet{vershynin2010introduction}, Prop 5.16 \footnote{In their notation and definitions, we use the fact that the corresponding $a_i$'s here are 1 and that $K=\max _i\left\|X_i\right\|_{\psi_1}=1$ as $\operatorname{max}V_{c}(w)-\mathbb{E}\left[ V_{c}(w)\right]=1$.}], we have
$$
\mathbb{P}\left(\sum_{c \in[s]}\left(V_{c}(w)-\mathbb{E}\left[ V_{c}(w)\right]\right)>t\right) \leq \exp \left(-C_{2} \min \left\{\frac{t^{2}}{s}, t\right\}\right),
$$
for some universal constant $C_{2}>0$. Applying union bound and using the fact that $\log \left|\mathcal{W}_{n, s}\right| \leq n\operatorname{log}(\operatorname{max}(k,s))$, we have
$$
\mathbb{P}\left(\max _{w \in \mathcal{W}_{n, s}} \sum_{c \in[s]}\left(V_{c}(w)-\mathbb{E}\left[ V_{c}(w)\right]\right)>t\right) \leq \exp \left(-C_{2} \min \left\{\frac{t^{2}}{s}, t\right\}+n \operatorname{log}(\operatorname{max}(k,s))\right).
$$
Thus, for any $C_{3}>0$, there exists $C_{4}>0$ only depending on $C_{2}$ and $C_{3}$, such that
\begin{align}
\label{4.1BoundV3}
\max _{w \in \mathcal{W}_{n, s}} \sum_{c \in[s]}\left(V_{c}(w)-\mathbb{E}\left[ V_{c}(w)\right]\right) \leq C_{3}\left(n \operatorname{log}(\operatorname{max}(k,s))+\sqrt{n s \operatorname{log}(\operatorname{max}(k,s))}\right),
\end{align}
with probability at least $1-\exp \left(-C_{4} n \operatorname{log}(\operatorname{max}(k,s))\right)$. Plugging the bounds (\ref{4.1BoundV2}) and (\ref{4.1BoundV3}) into (\ref{4.1BoundV}), we obtain
$$
\begin{aligned}
\sum_{ij}\left(\hat{\theta}_{ij}-\tilde{\theta}_{ij}\right)^{2} & \leq\left(C_{3}+C_{1}\right)\left(s+n \operatorname{log}(\operatorname{max}(k,s))+\sqrt{n s \operatorname{log}(\operatorname{max}(k,s))}\right) \\
& \leq 2\left(C_{3}+C_{1}\right)\left(s+n \operatorname{log}(\operatorname{max}(k,s))\right),
\end{aligned}
$$
with probability at least $1-\exp \left(-C_{4} n \operatorname{log}(\operatorname{max}(k,s))\right)$. The proof is complete.
\end{proof}

\section{Stochastic shape approximation}

\begin{proof}[Proof of Theorem 3.1]

Let us first give an outline of the proof of Theorem 3.1. In the definition of the class $\Theta_{s}$, we denote the true value on each shape by $\left\{Q_{c}^{*}\right\} \in[0,1]^{s}$ and the oracle assignment by $w^{*} \in \mathcal{W}_{n, s}$ such that $\theta_{ij}=Q_{w^{*}(i,j)}^{*}$ for any $(i,j)\in [n]^2$. To facilitate the proof, we introduce the following notation. For the estimated $\hat{w}$, define $\left\{\tilde{Q}_{c}\right\} \in[0,1]^{s}$ by $\tilde{Q}_{c}=\bar{\theta}_{c}(\hat{w})$, and also define $\tilde{\theta}_{ij}=\tilde{Q}_{\hat{w}(i,j)}$ for any $(i,j)\in [n]^2$.
Recall for any optimiser of the objective function,
$$
(\hat{Q}, \hat{w}) \in \underset{Q \in \mathbb{R}^{k}, w \in \mathcal{W}_{n, s}}{\operatorname{argmin}} L(Q, w).
$$
By the definition of this estimator, we have
\begin{equation}
L(\hat{Q}, \hat{w}) \leq L\left(Q^{*}, w^{*}\right),
\end{equation}
which can be rewritten as
\begin{equation}
\label{4.1}
\|\hat{\theta}-A\|^{2} \leq\|\theta-A\|^{2} .
\end{equation}
As in \citet{gao2015rate}, the left-hand side of  can be decomposed as
\begin{equation}
\label{4.2}
\|\hat{\theta}-\theta\|^{2}+2\langle\hat{\theta}-\theta, \theta-A\rangle+\|\theta-A\|^{2}.
\end{equation}
Combining (\ref{4.1}) and (\ref{4.2}), we have
\begin{equation}
\label{4.3}
\|\hat{\theta}-\theta\|^{2} \leq 2\langle\hat{\theta}-\theta, A-\theta\rangle .
\end{equation}
The right-hand side of (\ref{4.3}) can be bounded as
\begin{align}
\langle\hat{\theta}-\theta, A-\theta\rangle=&\langle\hat{\theta}-\tilde{\theta}, A-\theta\rangle+\langle\tilde{\theta}-\theta, A-\theta\rangle \nonumber \\
\leq &\|\hat{\theta}-\tilde{\theta}\|\left|\left\langle\frac{\hat{\theta}-\tilde{\theta}}{\|\hat{\theta}-\tilde{\theta}\|}, A-\theta\right\rangle\right| \label{4.4}\\
&+(\|\tilde{\theta}-\hat{\theta}\|+\|\hat{\theta}-\theta\|)\left|\left\langle\frac{\tilde{\theta}-\theta}{\|\tilde{\theta}-\theta\|}, A-\theta\right\rangle\right|\label{4.5}.
\end{align}
Using Lemmas B.3,B.2 and B.4 on the following three terms:
\begin{equation}
\label{4.6}
\|\hat{\theta}-\tilde{\theta}\|, \quad\left|\left\langle\frac{\hat{\theta}-\tilde{\theta}}{\|\hat{\theta}-\tilde{\theta}\|}, A-\theta\right\rangle\right|, \quad\left|\left\langle\frac{\tilde{\theta}-\theta}{\|\tilde{\theta}-\theta\|}, A-\theta\right\rangle \right|,    
\end{equation}
they can all be bounded by $C \sqrt{s+n \operatorname{log}(\operatorname{max}(k,s))}$ with probability at least
$$
1-3 \exp \left(-C^{\prime} n \operatorname{log}(\operatorname{max}(k,s))\right).
$$
Combining these bounds with (\ref{4.4}), (\ref{4.5}) and (\ref{4.3}) using bounds on (\ref{4.6}), we get
$$
\|\hat{\theta}-\theta\|^{2} \leq 2 C\|\hat{\theta}-\theta\| \sqrt{s+n \operatorname{log}(\operatorname{max}(k,s))}+4 C^{2}\left(s+n \operatorname{log}(\operatorname{max}(k,s))\right).
$$
Solving the above for $\|\hat{\theta}-\theta\|$ gives
$$ \|\hat{\theta}-\theta\|\leq C\sqrt{s+n \operatorname{log}(\operatorname{max}(k,s))} \left(1+\sqrt{5}\right),$$
which further gives
$$
\|\hat{\theta}-\theta\|^{2} \leq C_{1}\left(s+n \operatorname{log}(\operatorname{max}(k,s))\right),
$$
with probability at least $1-3 \exp \left(-C^{\prime} n \operatorname{log}(\operatorname{max}(k,s))\right)$, proving the high probability bound. To get the bound in expectation, we use the following inequality:
$$
\begin{aligned}
\mathbb{E}\left[n^{-2} \|\hat{\theta}-\theta\|^{2}\right]\leq & \mathbb{E}\left[n^{-2}\|\hat{\theta}-\theta\|^{2} \mathbb{I}\left\{n^{-2}\|\hat{\theta}-\theta\|^{2} \leq \varepsilon^{2}\right\}\right] \\
&+\mathbb{E}\left[n^{-2}\|\hat{\theta}-\theta\|^{2} \mathbb{I}\left\{n^{-2}\|\hat{\theta}-\theta\|^{2}>\varepsilon^{2}\right\}\right] \\
\leq & \varepsilon^{2}+\mathbb{P}\left(n^{-2}\|\hat{\theta}-\theta\|^{2}>\varepsilon^{2}\right) \leq \varepsilon^{2}+3 \exp \left(-C^{\prime} n \operatorname{log}(\operatorname{max}(k,s))\right),
\end{aligned}
$$
where $\varepsilon^{2}=C_{1}\left(\frac{s}{n^2}+\frac{\operatorname{log}(\operatorname{max}(k,s))}{n}\right)$. Since $\varepsilon^{2}$ is the dominating term, the proof is complete.
\end{proof}

\section{Graphon estimation through stochastic shape}

\begin{proof}[Proof of Lemma 3.1]
Define $w^{*}:[n]\times [n] \rightarrow[s]$ by
$$
\left(w^{*}\right)^{-1}(c)=\left\{(i,j) \in[n]\times[n]: (\xi_{i}, \xi_j) \in S_{c}\right\}.
$$
We use the notation $n_{c}^{*}=\left|\left(w^{*}\right)^{-1}(c)\right|$ for each $c \in[s]$ and $w_{c}^{*}=\left\{(u,v): w^{*}(u,v)=c\right\}$ for $c \in[s]$. By such construction of $w^{*}$, for $(i,j)$ such that $(\xi_{i}, \xi_j) \in S_c$ with $c \in [s]$, we have
$$
\begin{aligned}
&\left|f\left(\xi_{i},\xi_{j}\right)-\bar{\bar{\theta}}_{c}\left(w^{*}\right)\right| \\
=&\left|f\left(\xi_{i},\xi_{j}\right)-\frac{1}{n_{c}^{*}} \sum_{(u,v) \in w_{c}^{*}} f\left(\xi_{u},\xi_{v}\right)\right| \\
\leq & \frac{1}{n_{c}^{*}} \sum_{(u,v) \in w_{c}^{*}}\left|f\left(\xi_{i},\xi_{j}\right)-f\left(\xi_{u},\xi_{v}\right)\right|,\\
\leq & \frac{1}{n_{c}^{*}} \sum_{(u,v) \in w_{c}^{*}} M\left(\left|\xi_{i}-\xi_{u}\right|+\left|\xi_{j}-\xi_{v}\right|\right)^{\alpha \wedge 1}, \text{ by \eqref{holder_cond},} \\
\leq & C_{1} M D_{w^*}^{\alpha \wedge 1}, \text{ by \eqref{diamBeta},}\\
\leq & C_{2} M s^{-\frac{\beta}{2}(\alpha \wedge 1)}.
\end{aligned}
$$
The $\alpha\wedge 1$ arises because when $\alpha>1$, any function $f \in \mathcal{H}_{\alpha}(M)$ satisfies the Hölder condition for $\alpha=1$. Squaring the inequality and summing over $c \in[s]$ completes the proof.
\end{proof}

\begin{proof}[Proof of Theorem 3.3]
Taking the same argument as the proof of theorem 5.2 and the ideas from \citep{gao2015rate}, we obtain the following bound
$$
\left\|\hat{\theta}-\theta^*\right\|^2 \leq 2\left\langle\hat{\theta}-\theta^*, A-\theta^*\right\rangle,
$$
whose right-hand side can be bounded as
$$
\begin{aligned}
\left\langle\hat{\theta}-\theta^*, A-\theta^*\right\rangle &=  \langle\hat{\theta}-\tilde{\theta}, A-\theta\rangle+\left\langle\tilde{\theta}-\theta^*, A-\theta\right\rangle+\left\langle\hat{\theta}-\theta^*, \theta-\theta^*\right\rangle \\
&\leq  \|\hat{\theta}-\tilde{\theta}\|\left|\left\langle\frac{\hat{\theta}-\tilde{\theta}}{\|\hat{\theta}-\tilde{\theta}\|}, A-\theta\right\rangle\right|+\left(\|\tilde{\theta}-\hat{\theta}\|+\|\hat{\theta}-\theta^*\|\right)\left|\left\langle \frac{\tilde{\theta}-\theta^*}{\|\tilde{\theta}-\theta^*\|}, A-\theta\right\rangle \right| \\
& \quad+\|\hat{\theta}-\theta^*\|\left\|\theta-\theta^*\right\| .
\end{aligned}
$$
Now, using a similar method and notation as \citep{gao2015rate}, we set
    $$
\begin{gathered}
L=\left\|\hat{\theta}-\theta^*\right\|, \quad R=\|\tilde{\theta}-\hat{\theta}\|, \quad B=\left\|\theta-\theta^*\right\|, \\
E=\left|\left\langle\frac{\hat{\theta}-\tilde{\theta}}{\|\hat{\theta}-\tilde{\theta}\|}, A-\theta\right\rangle\right|, \quad F=\left|\left\langle\frac{\tilde{\theta}-\theta^*}{\left\|\tilde{\theta}-\theta^*\right\|}, A-\theta\right\rangle\right| .
\end{gathered}
$$
Then, by the derived inequalities, we have
$$
L^2 \leq 2 R E+2(L+R) F+2 L B .
$$
It can be rearranged as
$$
L^2 \leq 2(F+B) L+2(E+F) R .
$$
By solving this quadratic inequality of $L$, we can get
\begin{align}
\label{4.7}
L^2 \leq \max \left\{16(F+B)^2, 4 R(E+F)\right\} .
\end{align}
By Lemma B.2-B.4 and Lemma 3.1 in the main paper, for any constant $C^{\prime}>0$, there exist constants $C$ only depending on $C^{\prime}, M$, such that
$$
\begin{aligned}
B^2 & \leq C n^2\left(\frac{1}{s^{\beta}}\right)^{\alpha \wedge 1}, & F^2 & \leq C n \operatorname{log}(\operatorname{max}(k,s)), \\
R^2 & \leq C\left(s+n \operatorname{log}(\operatorname{max}(k,s))\right), & E^2 & \leq C\left(s+n \operatorname{log}(\operatorname{max}(k,s))\right),
\end{aligned}
$$
with probability at least $1-\exp \left(-C^{\prime} n\right)$. By \eqref{4.7}, we have
$$
L^2 \leq C_1\left(n^2\left(\frac{1}{s^\beta}\right)^{\alpha \wedge 1}+s+n \operatorname{log}(\operatorname{max}(k,s))\right),
$$
with probability at least $1-\exp \left(-C^{\prime} n\right)$ for some constant $C_1$. Hence, there is some constant $C_2$ such that
\begin{align}
\frac{1}{n^2} \sum_{i j}\left(\hat{\theta}_{i j}-\theta_{i j}\right)^2 & \leq \frac{2}{n^2}\left(L^2+B^2\right)\nonumber \\
& \leq C_2\left(\left(\frac{1}{s^\beta}\right)^{\alpha \wedge 1}+\frac{s}{n^2}+\frac{\operatorname{log}(\operatorname{max}(k,s))}{n}\right).\label{convGraphonDepS}
\end{align}
For a 2-dimensional function that is a graphon with $n^2$ observations, the classical non-parametric rate is $n^{-\frac{2\alpha}{\alpha+1}}$. Let $s(n)=n^\delta$, then $\delta = \frac{2\beta^{-1}}{\alpha+1}$.

\end{proof}

\bibliographystyle{apalike}%put agsm if you want name ref or plain for numbers, apalike or plain for ArxivPictures
\bibliography{ref}

\end{document}